\documentclass[twocolumn,showpacs,superscriptaddress,twoside,pra,aps,10pt]{revtex4-1}

\usepackage{amsmath}
\usepackage{amsthm}
\usepackage{amssymb}
\usepackage{mathtools}
\usepackage{dsfont}
\usepackage{cases}
\usepackage[vcentermath]{youngtab}
\usepackage[caption=false,subrefformat=parens,labelformat=parens]{subfig}

\usepackage{ifthen}
\usepackage{hyperref}

\newcommand{\ket}[1]{|#1\rangle}
\newcommand{\bra}[1]{\langle #1|}
\newcommand{\e}[0]{\mathrm{e}}
\renewcommand{\iota}[0]{\mathrm{i}}

\newcommand{\dw}[1]{\mathrm{d}#1\omega}
\newcommand{\identity}{\mathds{1}}
\newcommand{\set}[1]{\left\{#1\right\}}

\renewcommand{\vec}[1]{\boldsymbol{#1}}
\newcommand{\defeq}{\coloneqq}
\newcommand{\imm}[1]{\mathrm{imm}^{\Yboxdim{5pt}\yng(#1)}}
\newcommand{\coeff}[3][\vec{\tau}]{\alpha_{#3}^{\Yboxdim{5pt}\yng(#2)}(#1)}
\newcommand{\dispyng}[1]{\Yboxdim{8pt}\yng(#1)}
\newcommand{\inlineyng}[1]{\Yboxdim{6pt}\yng(#1)}
\newcommand{\supyng}[1]{\Yboxdim{5pt}\yng(#1)}

\newtheorem{theorem}{Proposition}

\begin{document}
\title{Permutational symmetries for coincidence rates in~multi-mode~multi-photonic~interferometry}
\author{Abdullah Khalid}
\email{akhali@ucalgary.ca}
\affiliation{Institute for Quantum Science and Technology, University of Calgary, Calgary, Alberta, T2N 1N4, Canada}
\author{Dylan Spivak}
\email{djspivak@lakeheadu.ca}
\affiliation{Department of Physics, Lakehead University, Thunder Bay, Ontario P7B 5E1, Canada}
\author{Barry C. Sanders}
\email{sandersb@ucalgary.ca}
\affiliation{Institute for Quantum Science and Technology, University of Calgary, Calgary, Alberta, T2N 1N4, Canada}
\affiliation{Program in Quantum Information Science, Canadian Institute for Advanced Research, Toronto, Ontario M5G 1M1, Canada}
\author{Hubert de Guise}
\email{hubert.deguise@lakeheadu.ca}
\affiliation{Department of Physics, Lakehead University, Thunder Bay, Ontario P7B 5E1, Canada}

\begin{abstract}
  We obtain coincidence rates for passive optical interferometry by exploiting the permutational symmetries of partially distinguishable input photons, and our approach elucidates qualitative features of multi-photon coincidence landscapes. We treat the interferometer input as a product state of any number of photons in each input mode with photons distinguished by their arrival time. Detectors at the output of the interferometer count photons from each output mode over a long integration time. We generalize and prove the claim of Tillmann et al.\ [\emph{Phys.~Rev.~X} {\bf 5} 041015 (2015)] that coincidence rates can be elegantly expressed in terms of immanants. Immanants are functions of matrices that exhibit permutational symmetries and the immanants appearing in our coincidence-rate expressions share permutational symmetries with the input state. Our results are obtained by employing representation theory of the symmetric group to analyze systems of arbitrary number of photons in arbitrarily sized interferometers.
\end{abstract}

\maketitle

\section{Introduction}
Passive optical interferometry is the study of multiple photons interfering in passive interferometers, which are optical devices that can realized by a combination of beamsplitters and phase shifters~\cite{Reck1994}. The distinguishability of the photons plays a crucial role in determining their interference, first studied theoretically by Fearn and Loudon~\cite{Fearn1987,Fearn1988} and experimentally by Hong, Ou and Mandel (HOM)~\cite{Hong1987} for a system of two photons. The interference of many partially distinguishable photons has been of interest recently following the proposal of the BosonSampling problem~\cite{Aaronson2013} and historically due to the interest in linear optical quantum computing~\cite{Knill2001}. These quantum information tasks require the interference of perfectly indistinguishable photons, which state-of-the-art experiments are unable to deliver~\cite{Spring2013,Broome2013,Tillmann2013,Spagnolo2014,Bentivegna2015,Tillmann2015}. Therefore, it is important to characterize the interference of partially distinguishable photons to understand the relevance of experimental implementations of quantum information tasks.

A number of recent studies provide theories of multi-photonic passive interferometry experiments to elucidate qualitative features of optical systems~\cite{Tamma2015,Shchesnovich2015b,Tichy2015,Tillmann2015,Shchesnovich2017}. One research direction has been to determine whether the interference of partially distinguishable photons yields a probability distribution that is close to the probability distribution pertinent to the BosonSampling problem ~\cite{Shchesnovich2014,Shchesnovich2015}. Others have investigated the bunching behavior of partially distinguishable bosons~\cite{Tichy2015} and whether such behavior can be exploited to create hypothetical particles that possess arbitrary permutational symmetries~\cite{Tichy2015}. It has been suggested that pairwise distinguishability may not be sufficient to characterize the interference of many photons~\cite{Menssen2017}. Partial distinguishability can be controlled via degrees of freedom such as path, time-of-arrival at the interferometer or polarization, and control via polarization can be achieved by doubling the number of paths instead~\cite{Dhand2015}.

We aim to build upon a recently proposed description of multi-photonic interferometry experiments, that relates the interference of photons to the symmetries of their state~\cite{SiHui2013,deGuise2014,Tillmann2015}. For instance, the HOM effect occurs because two perfectly indistinguishable photons maximally interfere because their state is symmetric in the exchange of the photons, but two partially indistinguishable photons partially interfere because their state is a linear combination of the fully-symmetric and anti-symmetric basis states. The work we build upon analyzes multi-photonic systems using the representation theory of the symmetric group, where previous authors limit themselves to systems with exactly one photon in each input and output mode. They show that by exploiting the permutational symmetries of multi-photon states the photonic Hilbert space decomposes into subspaces, where elements in each subspace possess known permutational symmetries~\cite{Chen1989}. Distinguishability of the input photons determines which subspaces their state lives in. For instance, if photons are fully indistinguishable, then, because they are bosons, their state is fully symmetric under any permutation of the photons. As their distinguishability increases, the multi-photon state extends into other subspaces.

Tillmann {\it et al.}~\cite{Tillmann2015} have advanced that the coincidence rate can be expressed in closed form in terms of immanants of the scattering matrix. Immanants are matrix functions constructed from the symmetric group; some of them, including the permanent, are known to be hard to compute~\cite{Scheel2004,Burgisser2000,Brylinski2003}. Specifically, coincidence rates are sums of squared moduli of linear combinations of immanants where the coefficients in the sum are functions of the photonic distinguishability. Here we assume that photons are only distinguishable by their time-of-arrival at the interferometer.

Here we present major extensions to the results of Refs.~\cite{SiHui2013,deGuise2014,Tillmann2015}. We discuss arbitrary input and output configuration of partially distinguishable photons, i.e.\ possibly multiple photons in each input and output mode. We employ the term input configuration to refer to fixed number of photons in each input mode and the term output configuration to refer to a fixed number of counts at detectors located in each output mode. We use symmetry arguments to show that, if there are multiple photons in some input or output modes, then subspaces possessing certain symmetries do not appear in the Hilbert space decomposition. 

One of the aforementioned works~\cite{Tillmann2015} provides a procedure to calculate coincidence rates for arbitrary total number of photons. Crucial to the procedure is the factorization of the coincidence rate into a matrix product~\cite{deGuise2014,Tillmann2015}. One of these matrices, called the rate matrix, features in other partial indistinguishability theories~\cite{Shchesnovich2014,Shchesnovich2015b,Tichy2015}. The elegant decomposition of the rate into a matrix product provides insights into the symmetry of the coincidence rate because the symmetric group links partial distinguishability to representations, and allows for the expression of the rates in terms of immanants. Unfortunately, it is not shown that the Tillman {\it et al.}~\cite{Tillmann2015} procedure is correct and works in all cases. 

Here we fill this gap and extend their results by providing procedures to understand coincidence rates at the output modes of a passive interferometer for any given input and output photon configuration, and proving that these procedures work. We show that for any input and output configuration, the coincidence rate can always be factored. Furthermore, we prove that the rate matrix carries a representation of the symmetric group, and hence can be block-diagonalized by standard methods~\cite{Chen1989}, leading to the expression of the coincidence rates in terms of immanants. An example where this extension is immediately useful is given in Sec.~\ref{sec:hom}, where we analyze corrections to the HOM effect.

We use symmetry arguments to show, in Sec.~\ref{sec:hom} and more generally, that, if there are multiple photons in some input or output modes, then subspaces possessing certain symmetries do not appear in the Hilbert space decomposition. Relevant computational details for the problem of four photons in two modes are found in Appendix~\ref{sec:appendixexample}. Hence, given knowledge of the distinguishability of the input photons, we are able to identify which immanants appear in the coincidence rate expression using physical arguments alone. Finally, we discuss, again in Appendix~\ref{sec:appendixexample}, how the problem reducing the rate matrix to block diagonal form can be cast as an eigenvalue problem using class operators. This implies in particular that when some subspaces are known to be excluded by symmetry arguments, one can better focus the computational resources on the relevant subset of eigenvectors and related immanants.

We also show in Secs.~\ref{sec:permutedphotons} and \ref{sec:ratesimmanants} that, under a rearrangement of photons at the input of the interferometer, the coincidence rates at the output of the interferometer are covariant (but generally not invariant): the rates are expressed in terms of the same immanants but with different coefficients. We show these coefficients transform linearly under the rearrangement of the input photons. Therefore, once the rate for one input is known, the rates for any permuted input is easy to calculate. Finally, we discuss in Sec.~\ref{sec:permutedmodes} the transformation of rates under permutation of modes.

Our formalism thus provides an intuitive qualitative understanding of how the permutational symmetries of distinguishable photons determine their interference. Our approach can be used to better understand the effect of partial distinguishability in modeling passive interferometry experiments aiming to perform quantum information tasks such as linear optical quantum computing~\cite{Knill2001} and BosonSampling~\cite{Aaronson2013,Spring2013,Broome2013,Tillmann2013,Spagnolo2014,Bentivegna2015,Tillmann2015}. By exploiting symmetries, we can also reduce some calculational tasks and provide additional insights into relations between the coincidence rates in various situations.

We aim to focus on practical aspects of our methods, therefore, we first present our results for an example of four photons interfering in a three-mode interferometer in Sec.~\ref{sec:example}. Our second example, discussed in Sec.~\ref{sec:hom}, is a study of the HOM experiment~\cite{Hong1987} for which the sources produce more than one photon in each input mode. Technical details and derivations of these examples have been postponed to appendices, which can be pursued by the avid reader.  In Sec.~\ref{sec:arbitraryrates} we generalize these results for any multi-mode multi-photon inputs and outputs. We summarize our results and provide a conclusion in Sec.~\ref{ref:conclusion}.

\section{Four photons in a three-mode interferometer}
\label{sec:example}

\begin{figure*}[!ht]
  \begin{center}
    \null\hfill
    \subfloat[]{\includegraphics{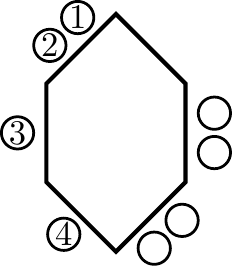}
      \label{fig:canonicalinputexample}}
    \hfill
    \subfloat[]{\includegraphics{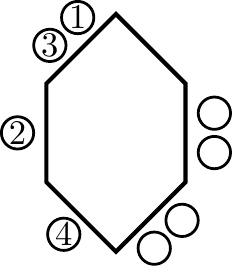}
      \label{fig:permutedphotonsexample}}
    \hfill
    \subfloat[]{\includegraphics{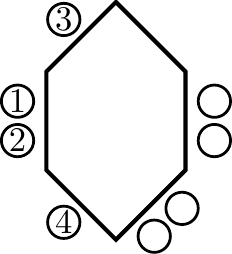}
      \label{fig:permutedmodesexample}}
    \hfill\null
  \end{center}
  \caption{An interferometer (hexagon) with three inputs (left edges) and three outputs (right edges). The numbered and unnumbered circles represent photons at the input and output ports. The photons at the input are numbered as their distinguishability can be controlled but those at the output are unlabeled as the detectors only count the number of photons in each output port. The input states in the three examples are (a) $\ket{211;1123;\vec{\tau}}$, (b) $\ket{211;1213;\vec{\tau}}$, and (c) $\ket{121;2213;\vec{\tau}}$.}
  \label{fig:examples}
\end{figure*}

Our main results can be appreciated by analyzing some relatively simple examples. Figure~\ref{fig:examples} shows three ways that four photons can interfere in a three-mode passive interferometer. We discuss how for such systems the Hilbert space decomposes into a direct sum of permutationally invariant subspaces. We show that the coincidence rates inherit these permutational symmetries. Finally, we explain how the rate expressions for these examples can be obtained from each other.

The examples shown in Fig.~\ref{fig:examples} can be described using four pieces of information: the path of the photons (encoded in two strings $\vec{\eta}$ and $\vec{\upsilon}$), the temporal state of the photons (encoded in the string $\vec{\tau}$), the measurement result (encoded in the string $\vec{\mu}$), and the action of the interferometer (encoded in the matrix $U$). We describe and discuss each of these now. 

First, we describe the configuration of the photons at the input of the interferometer. We label an input configuration by two strings. The first string $\vec{\eta}$, called the input \emph{mode-occupation} string, records the number of photons in each input mode. For the three examples $\vec{\eta}$ is $211$,$211$ and $121$, respectively. If the photons are all indistinguishable this is all that is needed to describe the configuration of the photons.

If the photons are distinguishable or partially distinguishable, more information is needed. The string $\vec{\upsilon}$, called the \emph{photon-occupation} string, stores in which mode each of the four photons are placed. For the first example, $\vec{\upsilon}$ is $1123$ because the first two photons are in the first mode while the third and fourth photons are in the second and third mode respectively. For the other two examples, $\vec{\upsilon}$ is $1213$ and $2213$ respectively. The string $\vec{\upsilon}$ contains all information needed to label the configuration of distinguishable photons. The string $\vec{\eta}$ can be found if $\vec{\upsilon}$ is known, but we specify both in order to simplify mathematical expressions that appear in this paper.

Photons have degrees of freedom besides path such as spectra, polarizations, and times of arrival at the interferometer. Therefore, to define a multi-mode multi-photon state we need to specify the labels of each of these degrees of freedom. We assume that the distinguishing degrees of freedom are under the experimentalist's control and the distinguishability of the photons with respect to each other can be controllably tuned. To keep the analysis simple we assume the photons arrive at different times at the interferometer but are otherwise identical. Therefore, a $4$-tuple $\vec{\tau}$ of arrival times is used to characterize the distinguishability of the photons, where  $\tau_i\in\mathbb{R}$ is the time-of-arrival of the $i$-th photon.

Relative time-of-arrival is effective for controlling distinguishability,
but photons can have other internal degrees of freedom such as polarization or orbital angular momentum, which needs another form of distinguishability control. However, polarization and other degrees of freedom are easily incorporated into scalar-field interferometry by converting the internal degree of freedom to a path degree of freedom~\cite{Dhand2015}. Mathematically, conversion of polarization to path follows the prescription that including
``combinations of polarizers and beam splitters, \dots{} one extends to $U(2n)\supset U(n)\times U(2)$'', as long as the same polarization transformation is applied to all paths~\cite{Rowe1999}. Polarization control has been employed for the recent demonstration of the triad phase for three-photon three-channel interferometry~\cite{Menssen2017}. Our framework and expression found in this paper can be generalized to multiple degrees of freedom in a straightforward manner. In summary, we represent multi-photon states as $\ket{\vec{\eta};\vec{\upsilon};\vec{\tau}}$. The input states of the three examples shown in Fig.~\ref{fig:examples} are $\ket{211;1123;\vec{\tau}}$, $\ket{211;1213;\vec{\tau}}$ and $\ket{121;2213;\vec{\tau}}$, respectively.

This labeling of input states reveals that these inputs are related to each other by permutations. The input shown in Fig.~\subref*{fig:permutedphotonsexample} is related to the canonical input (Fig.~\subref*{fig:canonicalinputexample}) by a permutation of the photons. Specifically, the second and third photons have been swapped, which is apparent by comparing $\vec{\upsilon} = 1123$ and $\vec{\upsilon}'=1213$. The input shown in the third example Fig.~\subref*{fig:permutedmodesexample} is related to the canonical input by a permutation of the input modes: the first and second modes have been swapped, which can be determined by comparing their photon-occupation strings. As permutations can be implemented by linear transformations on the photons or modes, for a fixed output event the coincidence rates for two permutations can be obtained from each other by linear transformations.

In order to study the permutation of photons and modes formally, we require the symmetric group $S_n$, which is the set of all possible permutations of $n$ objects. The objects can be the $n$ photons in the system, or the $n$ possible modes of the system. The action of this group on photons and modes captures the types of permutations for which Figs.~\subref*{fig:permutedphotonsexample}- \subref{fig:permutedmodesexample} are examples. For four photons in three modes, the group $S_4$ is used to define the action $P_\sigma$ for $\sigma \in S_4$ and the group $S_3$ is used to define the action $Q_{\sigma'}$ for $\sigma'$ in $S_3$. While we formally define these action in Sec.~\ref{sec:arbitraryrates}, here we note that
\begin{align}
  P_{(23)}\ket{211;1123;\vec{\tau}} &\defeq \ket{211;P_{(23)}1123;\vec{\tau}} \nonumber\\&= \ket{211;1213;\vec{\tau}}, \\
  Q_{(12)}\ket{211;1123;\vec{\tau}} &\defeq \ket{Q_{(12)}211;Q_{(12)}1123;\vec{\tau}} \nonumber \\&= \ket{121;2213;\vec{\tau}}. 
\end{align}

The next aspect of the passive interferometry experiment to be described is the interferometer. An interferometer is characterized by a matrix $U$ that specifies the linear transformation undergone by the input state. This matrix might be unitary, i.e.\ $UU^\dagger = U^\dagger U=\identity$, but is assumed to be a general linear complex matrix if the interferometer is lossy or if only $m$ input and output modes of a larger interferometer are used. For a three-mode interferometer shown in Fig.~\ref{fig:examples}, $U$ is a $3\times 3$ matrix that is identical for all three examples.

Finally, we outline the measurement scheme. We treat the detectors at the output of the interferometer as devices that count the total number of photons in a mode regardless of when they arrived at the detector. This lack of information about the paths of the photons through the interferometer generates the observed interference between the photon paths~\cite{Fearn1988}. This is why the photons at the outputs in Fig.~\ref{fig:examples} are not labeled. As we are interested in finding how the coincidence rate changes when the input state is modified, we fix the output to be same for all three examples shown in Fig.~\ref{fig:examples}. This output event is the one where photon detectors placed in the three output modes detect zero, two and two photons respectively, conveniently represented by the output mode-occupation string $\vec{\mu}=022$. The following analysis can be performed for any other output event by using the appropriate output mode-occupation string $\vec{\mu}$. With our description of the components of passive interferometry experiments complete, we can now turn to the task of calculating coincidence rates for the three examples shown in Fig.~\ref{fig:examples}. This is the content of the next three subsections.

\subsection{Rate for the canonical input}
\label{sec:mainexamplerate}
The expression and calculations of the coincidence rate in terms of immanants is a three-part procedure. First, the Hilbert space of the photonic systems is decomposed into permutationally-symmetric subspaces with the aid of an appropriately constructed basis-transformation matrix $V$. Next the coincidence rate is calculated in the form suggested in Refs.~\cite{deGuise2014,Tillmann2015}. In this way of calculating the coincidence rate $C(\vec{\tau})$, it is expressed as
\begin{equation}
  C(\vec{\tau}) = \vec{u}^\dagger R(\vec{\tau}) \vec{u}, \label{eq:rateform}
\end{equation}
where $\vec{u}$ is a vector, called the \emph{interferometer} vector, which depends only on the properties of the interferometer, whereas the \emph{rate} matrix $R$ captures the effects of the time-of-arrival of the photons on the coincidence rate. For the examples in Fig.~\ref{fig:examples}, $\vec{u}$ and $R$ are of dimension $4!/2!2!0! = 6$. The components of $\vec{u}$ are polynomials in the entries of the matrix $U$ and given explicitly for this case in Eq.~\eqref{eq:interferometervectorexample}. Additional properties of $\vec{u}$ and $R$ are described in Sec.~\ref{sec:arbitraryrates}. Finally, $\vec{u}$ and $R$ are rotated by $V$, which transforms the entries of $\vec{u}$ into sums of immanants and $R$ into a block-diagonal matrix. Hence, using Eq.~\eqref{eq:rateform} the coincidence rate may be expressed in terms of immanants.

In this subsection, we only give a rough outline of how to construct $V$, and in Appendix~\ref{sec:appendixexample} work out fully the $V$ for the examples of this section. The reader is directed to a textbook such as Ref.~\cite{Chen1989} for further details. Similarly, exact expressions for $\vec{u}$ and $R$ are delayed to Appendix~\ref{sec:appendixexample} because similar expressions have been derived in other places~\cite{SiHui2013,deGuise2014,Shchesnovich2014,Shchesnovich2015b,Tichy2015,Tillmann2015}, but we do provide general expressions in Sec.~\ref{sec:arbitraryrates}. Our main concern here is to understand how the symmetries of the coincidence rate $C(\vec{\tau})$ depends on the input and output configuration, and the distinguishability of the photons.

We begin by decomposing the $3^4$-dimensional Hilbert space $H$ of photon mode-occupations of four photons in three modes. We are interested in decomposing this space into the direct sum of subspaces, where each subspace has known permutational symmetries. This decomposition is divided into two steps. First note that the full Hilbert space $H$ is composed of orthogonal subspaces, where each subspace $H^{\vec{\mu}}$ corresponds to a distinct mode occupation $\vec{\mu}$, i.e.\
\begin{equation}
  H = \bigoplus_{\stackrel{\mu_i\in \set{0,\dots,4},}{\sum_i \mu_i = 4}} H^{\vec{\mu}}. \label{eq:photonconfigurationspaces}
\end{equation}
One of these subspaces is $H^{022}$, which corresponds to the output event of Fig.~\ref{fig:examples}. These orthogonal subspaces can be further decomposed into permutationally symmetric subspaces. As the procedure is the same for all of them, we only discuss the decomposition of $H^{022}$. 

The $H^{022}$ subspace is $4!/(0!2!2!)=6$ dimensional. A natural basis for this subspace is spanned by the six distinct vectors obtained by applying the permutation operator $P_\sigma$ to $\ket{022;2233}$ for $\sigma\in S_4$. Formally we can express this basis as
\begin{equation}
  B = \set{P_\sigma\ket{022;2233;\vec{\tau}}: \sigma\in S_4}. \label{eq:naturalbasis}
\end{equation}
The action $P_\sigma$ of $S_4$ on this basis defines a six-dimensional reducible representation $\Gamma$ of $S_4$, the formal definition of which is provided in Sec.~\ref{sec:arbitraryrates}. In simple terms, given any element $\sigma \in S_4$, $\Gamma(\sigma)$ is the $6\times 6$ permutation matrix that exchanges the elements of $B$ according to how $P_\sigma$ acts on $B$. Using standard representation theory methods, the reducible representation $\Gamma$ can be decomposed into a direct sum of irreducible representations (irreps)~\cite{Chen1989} of $S_4$.

The irreps of $S_4$ are in one to one correspondence with the conjugacy classes of $S_4$, where each conjugacy class can be labeled by a partition of $4$. Partitions can be graphically depicted by Young diagrams, so each irrep of the permutation group can be labeled by Young diagrams. For $S_4$, the five possible $4$-box Young diagrams are $\inlineyng{4}, \inlineyng{3,1}, \inlineyng{2,2}, \inlineyng{2,1,1}, \inlineyng{1,1,1,1}$ (the diagram $\inlineyng{3,1}$ corresponds to the partition $3+1$ etc.). Each of these diagrams is a label for an irrep of $S_4$. The Young diagram corresponding to each irrep indicates the symmetries of the irrep, where the irrep is symmetric across each row and anti-symmetric across each column.

Using, standard representation theory methods~\cite{Chen1989}, $\Gamma$ decomposes as
\begin{equation}
  \Gamma = \dispyng{4} \oplus \dispyng{3,1} \oplus \dispyng{2,2}, \label{eq:exampledecomposition}
\end{equation}
where $\inlineyng{4}$ is one-dimensional, $\inlineyng{3,1}$ is three-dimensional and $\inlineyng{2,2}$ is two-dimensional. The other two irreps of $S_4$, $\inlineyng{2,1,1}$ and $\inlineyng{1,1,1,1}$, do not appear in the decomposition for our example. This is because the output event $\vec{\mu}=022$, requires symmetry in the placement of the third and fourth photons since they are in the same mode, but $\inlineyng{2,1,1}$ and $\inlineyng{1,1,1,1}$ are anti-symmetric in the third and fourth boxes. Further details on how to obtain this decomposition are provided in Eq.~\eqref{eq:221decomposition}. This decomposition, also called block-diagonalization, of $\Gamma$ can be realized by a matrix we label as $V$, explicitly provided in Eq.~\eqref{eq:basistransformation22}. This matrix is used to rotate $\vec{u}$ and $R$. 

The action of the permutation group commutes with the action of the unitary group on photonic states~\cite{weyl1950,brauer1973,Rowe2012}, i.e.\,
\begin{equation}
  [U,P_\sigma] = 0, \qquad \sigma \in S_4.
\end{equation}
Consequently, the subspace $H^{022}$ decomposes as
\begin{equation}
  H^{022} = H^{\supyng{4}} \oplus H^{\supyng{3,1}} \oplus H^{\supyng{2,2}},
\end{equation}
where, in an obvious notation, $H^{\supyng{4}}$ has the same permutational symmetries as $\inlineyng{4}$, etc. The permutational symmetries of these subspaces become clear when we discuss how the coincidence rate varies as the distinguishability of the photons is tuned. 

As mentioned in the introduction and later proved in Sec.~\ref{sec:arbitraryrates}, the coincidence rate at the output of the interferometer can be expressed in the form~\eqref{eq:rateform}. This is a generalization of the terminology and formalism of previous work which was restricted to one photon in each input and output mode~\cite{deGuise2014,Tillmann2015}. Expressing the rate in the form~\eqref{eq:rateform} is convenient for a number of reasons. First, this form separates the effects on the rate of the interferometer and the time-of-arrival of the photons. The interferometer vector $\vec{u}$ depends only on the entries of the interferometer matrix $U$ but not on the time-of-arrival of the photons. In the basis $B$~\eqref{eq:naturalbasis}, the entries of the interferometer vector are the amplitudes of the transitions from the input state $\ket{211;1123;\vec{\tau}}$ to the vectors $B$ associated with the output event $\vec{\mu}=022$. 

The rate matrix $R$ depends on the time-of-arrival of the photons, which determines their distinguishability. The entry $R_{ij}$ of the rate matrix is determined by how distinguishable are the $i$th and $j$th transitions whose amplitudes are recorded in the interferometer vector. The rate matrix is a square Hermitian matrix.

The rate matrix does not depend in any way on the input configuration of the photons. For a fixed output event such as $\vec{\mu}=022$, the rate matrix is the same irrespective of the input. Concretely, for all inputs shown in Fig.~\ref{fig:examples}, and for any other input of four photons which are not related by permutations to them (say $\ket{400;1111;\vec{\tau}}$), the rate matrix is the same. The effect of the input configuration on the rate expression is only felt through the interferometer vector. 

Our main result, proved in Sec.~\ref{sec:arbitraryrates}, is that the rate matrix carries a reducible representation $\Gamma$ of $S_4$. This means, that the rate matrix can be expressed as a linear combination of the representations of all $\sigma$ in $S_4$. Mathematically, we write
\begin{equation}
  R(\vec{\tau}) = \sum_{\sigma\in S_4} \Delta_\sigma(\vec{\tau})\Gamma(\sigma),
\end{equation}
where $\Gamma(\sigma)$ is the representation of $\sigma$ and $\Delta_\sigma(\vec{\tau})$ is a distinguishability-dependent coefficient. The expression of the rate matrix in this form implies that the basis $V$ that reduces $\Gamma$ also block-diagonalizes the rate matrix. In the block-diagonal form, each block of the rate matrix is associated with an irrep of $S_4$ and is of the same size as the dimension of the irrep.  The coincidence rate~\eqref{eq:rate211} is calculated by employing the formula
\begin{equation}
  C(\vec{\tau}) = \left(V\vec{u}\right)^\dagger\left[VR\left(\vec{\tau}\right)V^\dagger\right]\left(V\vec{u}\right). \label{eq:rotatedrate}
\end{equation}

The elements of $V\vec{u}$ are sums of immanants of the scattering matrix $T$, which here is a four-dimensional matrix formed by taking the $\vec{\upsilon}=1123$ rows (input) and $\vec{\xi}=2233$ columns (output) of the interferometer matrix $U$; i.e.\ $T$ has entries
\begin{equation}
  T_{ij} \defeq U_{\upsilon_i,\xi_j}. \label{eq:scatteringmatrixdef}
\end{equation}
Immanants are matrix functions associated with the irreps of the symmetric group. The irrep labeled by Young diagram $\lambda$ is associated with the immanant
\begin{equation}
  \text{imm}^\lambda T \defeq \sum_{\sigma \in S_n}\chi^\lambda(\sigma)\left(\prod_{i=1}^nT_{i\sigma(i)}\right), \label{eq:immanantdef}
\end{equation}
where $\chi^\lambda(\sigma)$ is the character of the permutation $\sigma$ in the irrep $\lambda$.

We define $T_\sigma$ as the matrix $T$ with its rows permuted by $\sigma\in S_4$ (note $T = T_\e$). The set of immanants
\begin{equation}
  \set{\text{imm}^\lambda T_\sigma: \sigma \in S_4},
\end{equation}
has fewer distinct elements than $|S_4| = 24$ because the immanants for two different $\sigma$ may be be equal (up to a constant), i.e.\  its possible that
\begin{equation}
  \text{imm}^\lambda T_\sigma \propto \text{imm}^\lambda T_{\sigma'}.
\end{equation}
For instance, the permanent $\imm{4}T$ is fully symmetric under any permutation of the rows of $T$, i.e.\ for all $\sigma\in S_4$, we find that $\imm{4}T_\sigma = \imm{4}T$. On the other hand, for half the $\sigma\in S_4$, $\imm{3,1}T_\sigma$ is proportional to $\imm{3,1}T$, and for the other half $\imm{3,1}T_\sigma$ is proportional to $\imm{3,1}T_{(13)}$. Finally, due to the multiplicity of photons in input and output modes (repetition of the rows and columns of $T$), for all $\sigma \in S_4$, we find that $\imm{2,2}T_\sigma$ is proportional to $\imm{2,2}T$. Therefore, we can express the coincidence rate only in terms of the set of immanants
\begin{equation}
  \set{\imm{4}T, \imm{3,1}T, \imm{3,1}T_{(13)}, \imm{2,2}T}. \label{eq:examplesetofimmanants}
\end{equation}

We are now ready to express the coincidence rate for  input $\ket{211;1123;\vec{\tau}}$ and output $\vec{\mu}=022$. The expressions for $\vec{u}$, $R$ and $V$, presented in Appendix~\ref{sec:appendixexample}, are plugged into Eq.~\eqref{eq:rotatedrate}, yielding the rate
\begin{align}
  C(\vec{\tau}) =&\hspace{0.4em} \coeff{4}{\e,\e}\left|\imm{4}T\right|^2 + \coeff{3,1}{\e,\e}\left|\imm{3,1}T\right|^2 \nonumber\\&+ \coeff{3,1}{(13),(13)}\left|\imm{3,1}T_{(13)}\right|^2 \nonumber\\ &+ \left[\coeff{3,1}{\e,(13)}\left(\imm{3,1}T\right)^\dagger \imm{3,1}T_{(13)} + \text{cc}\right] \nonumber\\&+ \coeff{2,2}{\e,\e}\left|\imm{2,2}T_{(22)}\right|^2, \label{eq:rate211}
\end{align}
where the coefficients $\alpha^\lambda_{\sigma,\sigma'}(\vec{\tau})$ are labeled by the immanant to which they are attached, $\e$ is the identity element in $S_4$ and cc refers to the complex conjugate of the preceding terms. The $\alpha^\lambda_{\sigma,\sigma'}(\vec{\tau})$ are linear combinations of $\Delta_\sigma(\vec{\tau})$ and the full rate is presented in Eq.~\eqref{eq:examplerate}.

For the scattering matrix in this example there are two different immanants of the $\inlineyng{3,1}$ type and one immanant of the $\inlineyng{2,2}$ type. In general, there are three $\inlineyng{3,1}$ immanants so they can span the three-dimensional subspace labeled by $\inlineyng{3,1}$ in~\eqref{eq:exampledecomposition}. Similarly, in general there are two $\inlineyng{2,2}$ immanants. The full number of immanants does not occur as a result of the multiplicity of photons in the input and output modes. There are relatively simple rules for determining the number of immanants of each type that appear in the rate and these are presented in Sec.~\ref{sec:ratesimmanants}.

The coefficients $\alpha^\lambda_{\sigma,\sigma'}(\vec{\tau})$ in~\eqref{eq:rate211} have well-defined behavior for different regimes of the distinguishability of the photons. We treat these cases one by one. If the photons are all indistinguishable i.e.\ $\forall i,j~\tau_i=\tau_j$, then the coefficients take values such that the rate reduces to
\begin{equation}
  C(\vec{\tau}) = \left|\imm{4}T\right|^2.
\end{equation}
The rate depends only on the permanent of the scattering matrix because permuting indistinguishable photons cannot change the rate, and the permanent is the only immanant that does not change under any permutation of its rows.

If three photons are indistinguishable but the fourth is distinguishable from the three, then the rate depends on not only the permanent but also on the $\inlineyng{3,1}$ type immanants. In this regime, the $\coeff{2,2}{\e,\e}$ in~\eqref{eq:rate211} vanishes. The next case is when the photons are divided into two pairs such that the two photons in each pair are indistinguishable with respect to each other, but the two pairs are mutually distinguishable. In this case all immanants appear in the rate and the rate has its most general form of Eq.~\eqref{eq:rate211}.

It is possible to make the photons even more distinguishable.  A pair of photons can be mutually indistinguishable, but the pair is distinguishable with the third photon which is further distinguishable with respect to the fourth photon. An example coincidence landscape associated with this last scenario is plotted in Fig.~\ref{fig:landscape211}. The first two photons arrive simultaneously at a fixed time, while the third and fourth photons are separately delayed with respect to this pair. The point at the center of the landscape $\tau_3=\tau_4=0$ only depends on the permanent. To calculate the coincidence rate along the line $\tau_3 =0$, the line $\tau_4=0$ or the line $\tau_3=\tau_4$ requires the additional computation of $\inlineyng{3,1}$-type immanants. Any other point on the landscape requires computing all immanants appearing in Eq.~\eqref{eq:rate211}, as in this case none of the coefficients are zero.

The final scenario is one where all four photons are distinguishable with respect to each other. Calculating the coincidence rate at a general point in this case requires the entirety of Eq.~\eqref{eq:rate211}. 

\begin{figure}[t]
  \begin{center}
    \includegraphics[width=0.45\textwidth]{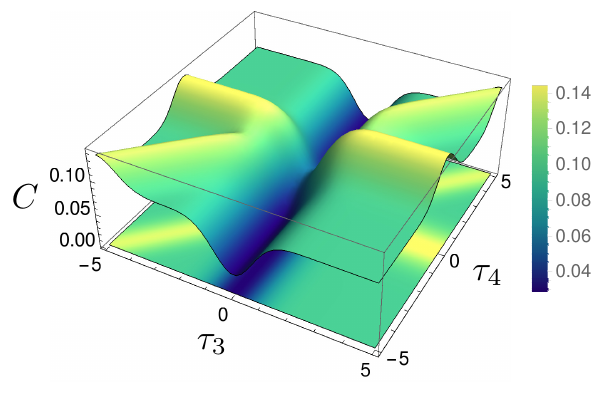}
  \end{center}
  \caption{The coincidence rate $C(\tau)$ for input state $\ket{211;1123;(0,0,\tau_3,\tau_4)}$ and output $\vec{\mu}=022$. The matrix chosen for this plot is provided in Eq.~\eqref{eq:plotmatrix}. The rate is shown both as a three-dimensional plot and as a contour plot at the base of the figure.}
  \label{fig:landscape211}
\end{figure}

We have calculated the coincidence rate for the input state $\ket{211;1123;\vec{\tau}}$ and output event $\vec{\mu}=022$ in terms of the three immanants of the scattering matrix. Using this coincidence expression we have shown that, for this input and output pair three regimes of distinguishability of the photons exist.

\subsection{Rate transformation under permutation of photons}
\label{sec:permutedphotons}
Now we turn to the problem of analyzing how the rate for a fixed output event changes under the permutation of the photons at the input. Fig.~\subref*{fig:permutedphotonsexample} shows an example where the photons at the interferometer input have been permuted with respect to the photons shown in Fig.~\subref*{fig:canonicalinputexample}. The second and third photons have been swapped so the second photon is now in the second input and the third photon is in the first input. This input state is represented by $\ket{211;1213}$.

\begin{figure*}[ht]
  \begin{center}
    \null\hfill
    \subfloat[]{\includegraphics{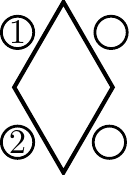}
      \label{fig:hom11}}
    \hfill
    \subfloat[]{\includegraphics{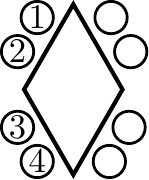}
      \label{fig:hom22to22}}
    \hfill
    \subfloat[]{\includegraphics{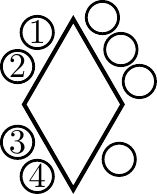}
      \label{fig:hom22to31}}
    \hfill
    \subfloat[]{\includegraphics{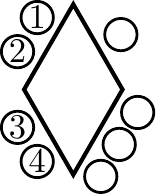}
      \label{fig:hom22to13}}
    \hfill\null
  \end{center}
  \caption{An interferometer (diamond) with two inputs (left edges) and two outputs (right edges). The numbered and unnumbered circles represent photons at the input and output ports. The inputs are (a) $\ket{11;12;\vec{\tau}}$, and (b), (c) and (d)  $\ket{22;1122;\vec{\tau}}$. The outputs are (a) $\vec{\mu}=11,\vec{\xi}=12$, (b) $\vec{\mu}=22,\vec{\xi}=1122$, (c) $\vec{\mu}=31,\vec{\xi}=1112$, and (d) $\vec{\mu}=13,\vec{\xi}=1222$}.
  \label{fig:hominputsoutputs}
\end{figure*}

As discussed above, the action $P_\sigma$ for $\sigma \in S_4$ captures permutations of the photons at the input of the interferometer. Though $S_4$ has $24$ elements there are only twelve distinct ways of placing four photons in the input ports with mode occupations $\vec{\eta}=211$. This is due to the multiplicity of photons in the first input port. The coincidence rates for any of these $12$ distinct permutations have the same form as~\eqref{eq:rate211}, except the coefficients $\alpha^\lambda_{\sigma,\sigma'}$ are different in each case. We now describe how coefficients for one input permutation are related to the coefficients for some other input permutation.

Under the action $P_\sigma$, the vector $\vec{u}$ transforms to $\Gamma(\sigma) \vec{u}$, while the rate matrix remains unaffected as it does not depend on the input configuration. Some straightforward matrix algebra results in the rate $C_\sigma(\vec{\tau})$ for the permuted input in the form
\begin{equation}
  C_\sigma(\vec{\tau}) = \left(V\vec{u}\right)^\dagger \left[V\Gamma\left(\sigma\right)^\dagger R\left(\vec{\tau}\right)\Gamma\left(\sigma\right)V^\dagger\right] \left(V\vec{u}\right). \label{eq:permutedphotons}
\end{equation}
The interferometer vector is unmodified for any $\sigma$ but the rate matrix is rotated by $V\Gamma(\sigma)^\dagger$. The expression~\eqref{eq:rate211} remains covariant. For any permutation of the input photons the rate expression is obtained by mixing the coefficients $\alpha^\lambda_{\sigma,\sigma'}$ linearly into each other.

\subsection{Rate transformation under permutation of modes}
\label{sec:permutedmodes}
The rate also remains covariant under a permutation of the modes. The example shown in Fig.~\subref*{fig:permutedmodesexample} has its input modes permuted with respect to the input modes of the example in Fig.~\subref*{fig:canonicalinputexample}. In total, there are three possible ways of permuting the modes, explicitly specified by the input mode-occupation strings $211$, $121$ and $112$. The action $Q_\sigma$ transforms between these various possible inputs.

The rate for each of these inputs has the form of~\eqref{eq:rate211} except for each input the scattering matrix $T$ is different. The matrix $T$ is a matrix formed by choosing the appropriate rows and columns of the interferometer matrix $U$. Therefore, it is straightforward to calculate the rate for any permutation of the modes once the form~\eqref{eq:rate211} is known.

\section{Hong-Ou-Mandel interference}
\label{sec:hom}

In this section we use our formalism to analyze the HOM effect for sources that sometimes produce two photons. The HOM effect is often used to demonstrate that photon sources are non-classical: classical fields interfering in a balanced beamsplitter have visibility less than half the visibility of two photons interfering similarly~\cite{Mandel1983}. Here we calculate corrections to the coincidence rate for the HOM effect when the sources sometimes produce two photons. We show that expressing the rate in terms of immanants makes it straightforward to calculate features of the coincidence rate.

The HOM setup consists of two photon sources, a beamsplitter and two photon detectors. Moreover, the experimentalist is able to tune the relative distinguishability of the photons created by the two photon sources. In the original HOM experiment~\cite{Hong1987}, the distinguishing degree of freedom was the overlap of the photons in the temporal domain, but it can also be the polarization, frequency or some other degree of freedom. The output fields of the photon sources are injected into the two input ports of the beamsplitter. The photon detectors are coupled to the output ports of the beamsplitter. They are usually bucket detectors that distinguish between zero photons and one or more photons. The output event of interest is the detection of photons in both the output ports of the beamsplitter. 

The photon sources in state-of-the-art experiments are parametric downconversion sources~\cite{Spring2013,Broome2013,Tillmann2013,Spagnolo2014,Bentivegna2015,Tillmann2015} which are non-deterministic photon sources with a photon pair-production probability we denote as $p$. These sources can simultaneously produce two pairs with probability $p^2$. These two possibilities and corresponding output events are shown in Fig.~\ref{fig:hominputsoutputs}.

When only one photon is produced in each input port $\vec{\eta}=11,\vec{\upsilon}=12$ and the output event of interest is $\vec{\mu}=11,\vec{\xi}=12$. This pair of events leads to the coincidence rate
\begin{align}
  C^{12,12}(\vec{\tau}) &= \coeff{2}{\e,\e}\left|\imm{2}T^{12,12}\right|^2 \nonumber \\&\quad + \coeff{1,1}{\e,\e}\left|\imm{1,1}T^{12,12}\right|^2, \label{eq:hom11to11}
\end{align}
where $\vec{\tau} = (\tau_1,\tau_2)$ and for clarity we label both the coincidence rate and the scattering matrix by $\vec{\upsilon}=12$ and $\vec{\xi}=12$. When the photons are indistinguishable ($\tau_1=\tau_2$) then the rate
\begin{equation}
  C^{12,12}(\tau_1,\tau_1) = \left|\imm{2}T^{12,12}\right|^2
\end{equation}
is equal to the modulus squared of the permanent.

When the sources produce two photons in each input port, the input event is labeled by $\vec{\eta}=22$ $\vec{\upsilon}=1122$. There are three possible output events labeled by the pairs $\vec{\mu}=22,\vec{\xi}=1122$; $\vec{\mu}=31,\vec{\xi}=1112$; and $\vec{\mu}=13,\vec{\xi}=1222$, such that there is at least one photon in each output mode (the others can be post-selected away). For all three of these output events the coincidence rate has the form
\begin{align}
  C^{1122,\xi}(\vec{\tau}') &= \coeff[\vec{\tau}']{4}{\e,\e}\left|\imm{4}T^{1122,\xi}\right|^2 \nonumber \\&\quad+ \coeff[\vec{\tau}']{3,1}{\e,\e}\left|\imm{3,1}T^{1122,\xi}\right|^2 \nonumber \\&\quad+ \coeff[\vec{\tau}']{2,2}{\e,\e}\left|\imm{2,2}T^{1122,\xi}\right|^2. \label{eq:hom22toxx}
\end{align}
As the first two input photons are in the first mode and the last two input photons are in the second mode, $\vec{\tau'} = (\tau_1,\tau_1,\tau_2,\tau_2)$.

The total coincidence rate is
\begin{align}
  C_{\text{total}} &= pC^{12,12}(\vec{\tau}) + p^2[C^{1122,1122}(\vec{\tau}') \nonumber\\&\quad+ C^{1122,1112}(\vec{\tau}') + C^{1122,1222}(\vec{\tau}')]. \label{eq:homtotal}
\end{align}
Exact expressions for these rates are presented in Appendix~\ref{sec:appendixhom}.

Figure~\ref{fig:homrate} depicts the total coincidence rate as well as the rate $C^{12,12}(\vec{\tau})$ as a function of $\tau_2$ for $\tau_1=0$. The presence of multi-photon effects increases the coincidence rate, but this increase depends on the distinguishability of the photons. Figure~\ref{fig:homrate} shows that the increase is least when the photons are completely indistinguishable and increases monotonically as the photons become more distinguishable. For indistinguishable photons the difference in rates
\begin{align}
  h_{\text{f}} &= p^2(\left|\imm{4}T^{1122,1122}\right|^2 + \left|\imm{4}T^{1122,1112}\right|^2 \nonumber\\&\qquad+ \left|\imm{4}T^{1122,1222}\right|^2),
\end{align}
is the sum of squared permanents. As the two pair of photons become more distinguishable, the other immanants appear in the coincidence rate and consequently the difference $h_{\text{c}}$ is greater than the difference $h_{\text{f}}$. The exact expression for $h_{\text{c}}$ is provided in Appendix~\ref{sec:appendixhom}.

We have shown that the shape of the HOM coincidence rate for non-deterministic photon sources differs qualitatively from the rate for perfect single-photon sources. We expect that interference experiments with higher number of photons in larger interferometers will yield similar qualitative differences if the inputs are contaminated with multiple photons.

\begin{figure}
  \centering
  \includegraphics[width=0.5\textwidth]{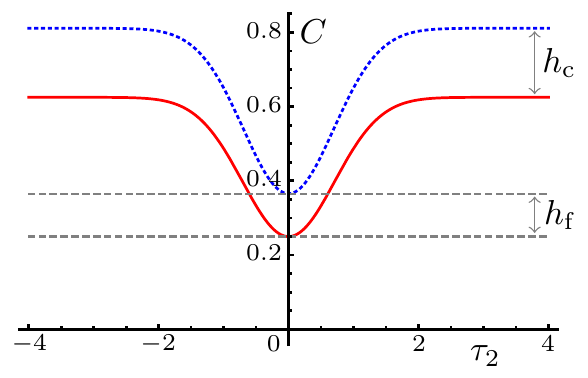}
  \caption{Coincidence rates with two different sources for a HOM experiment using for a $25:75$ beamsplitter. The solid (red) line shows the coincidence rate~\eqref{eq:hom11to11} for one photon in each input port. The dotted (blue) line shows the coincidence rate~\eqref{eq:homtotal} for up to two photons in each input port, with $p=0.04$. The rates are drawn as a function of $\tau_2$ while $\tau_1=0$. The difference between the two rates at $\tau_2=0$ is $h_{\text{f}}$ while the difference between the two rates at $\tau_2\to\infty$ is $h_{\text{c}}$.}
  \label{fig:homrate}
\end{figure}

\section{Rates for arbitrary input states}
\label{sec:arbitraryrates}
In this section, we show that for any passive interferometry experiment, the coincidence rate can always be expressed in the form of~\eqref{eq:rateform}. In Sec.~\ref{sec:passiveinterferometery} we formally define our model of passive interferometry experiments. In Sec.~\ref{sec:symmetricgroupaction} we define the action of the symmetric group on photonic states. With this background established we are able to calculate coincidence rates for any input and output in Sec.~\ref{sec:ratederivation}.

\subsection{Passive interferometry}
\label{sec:passiveinterferometery}
We define multi-mode multi-photon states, the action of the interferometer and the behavior of photon-counting detectors. Photonic states are defined using the monochromatic photonic creation and annihilation operators, $a_i(\omega_k),a_j^\dagger(\omega_l)$, for modes $i,j$ and frequencies $\omega_k,\omega_l$, satisfying the canonical commutation relations
\begin{equation}
\left[a_i(\omega_k),a_j^\dagger(\omega_l)\right] = \delta_{ij}\delta(\omega_k-\omega_l). \label{eq:commutationrelations}
\end{equation}
Realistic photons are not monochromatic, but rather have a frequency spectrum captured by a complex-valued function $\phi(\omega)$. A single photon in mode $i$ has state
\begin{equation}
  \ket{1,\tau}_i \defeq A_i^\dagger(\tau)\ket{0} \defeq \int \mathrm{d}\omega\phi(\omega)\e^{-\iota\omega\tau}a_i^\dagger(\omega)\ket{0}, \label{eq:bigadefinition} 
\end{equation}
where $A_i(\tau)$ is the creation operator of a photon that arrives at the interferometer at time $\tau$. 

The state of $n$ photons in $m$ modes is specified using the mode-occupation string, the photon-occupation string and the distinguishability vector as we did in previous sections. The input mode-occupations are specified by a string $\vec{\eta}$ of length $m$, where $\eta_i$ is the number of photons in mode $i\in\set{1,\dots,m}$ and $\sum \eta_i = n$. The photon-occupations are specified by a string $\vec{\upsilon}$ of length $n$, where $\upsilon_i$ is the mode occupied by the $i$-th photon for $i\in\set{1,\dots,n}$. The time-of-arrival of the photons at the interferometer is stored in the length $n$ vector $\vec{\tau}$.

Using $\vec{\eta}$, $\vec{\upsilon}$ and $\vec{\tau}$ we can define the state of $n$ photons in $m$ modes as
\begin{equation}
  \ket{\vec{\eta};\vec{\upsilon};\vec{\tau}} \defeq \frac{1}{\sqrt{\mathcal{N}}}\prod_{i=1}^nA_{\upsilon_i}^\dagger(\tau_i)\ket{0}, \label{eq:fockstate}
\end{equation}
where $\frac{1}{\sqrt{\mathcal{N}}}$ is the normalization of this wavefunction that we evaluate in the proof of Proposition~\ref{th:arbitraryrates} below.

The interferometer is described by an $m\times m$ matrix~$U$. The interferometer transforms the input creation operators to output creation operators according to
\begin{equation}
  a_{i,\text{in}}^\dagger(\omega_j) = \sum_{k=1}^m U_{ik}a_{k,\text{out}}^\dagger(\omega_j), \label{eq:interferometeraction} 
\end{equation}
where, in $U_{ik}$, the first index $i$ denotes the input mode and the second index $k$ denotes the output mode. Elsewhere in this paper, for typographical reasons, we suppress the subscripts $_{\text{in}}$ and $_{\text{out}}$ on the creation and annihilation operators as the difference is clear from context.

We model detectors as photon number counters that output the number of photons detected in a particular mode. For simplicity of presentation we assume that all detectors are identical, lossless at all frequencies and independent of the time-of-arrival of the photons. Much like the input mode-occupation string, the number of photons detected in each output mode are collected in an output mode-occupation string $\vec{\mu}$, where $\mu_i$ is the number of photons detected in the $i$th output mode and $\sum_i\mu_i = n$. The measurement operator $M_{\vec{\mu}}$ is defined using the length $n$ vector
\begin{equation*}
  \vec{\xi} \defeq (\overbrace{1,\dots ,1}^{\mu_1\text{ times}} ,\dots, \overbrace{i,\dots ,i}^{\mu_i\text{ times}} ,\dots, \overbrace{m,\dots ,m}^{\mu_m\text{ times}}).
\end{equation*}
Then the measurement operator element
\begin{align}
  M_{\vec{\mu}} \defeq & \frac{1}{\mu_1!,\cdots,\mu_n!}\int  \mathrm{d}\omega_1\cdots\mathrm{d}\omega_n |\varphi(\omega_1)|^2\cdots |\varphi(\omega_n)|^2\nonumber\\&a_{\xi_1}^\dagger(\omega_1)\cdots a_{\xi_n}^\dagger(\omega_n)\ket{0}\bra{0}a_{\xi_1}(\omega_1)\cdots a_{\xi_n}(\omega_n), \label{eq:measurementoperator}
\end{align}
where the $\varphi(\omega_i)$ specify the spectral range of the detectors. Usually, this spectral range is much broader than the spectral width of the photonic wavefunction~\cite{Eisaman2011,Broome2013}, and therefore everywhere we make the assumption that $\varphi(\omega_i)\phi(\omega_i) = \phi(\omega_i)$. If the state $\ket{\vec{\eta};\vec{\upsilon};\vec{\tau}}$ is measured, the probability $\rm{Pr}(\vec{\mu}|\vec{\eta};\vec{\upsilon};\vec{\tau})$ of outcome $\vec{\mu}$ is
\begin{equation}
  \mathrm{Pr}(\vec{\mu}|\vec{\eta};\vec{\upsilon};\vec{\tau}) = \bra{\vec{\eta};\vec{\upsilon};\vec{\tau}}M_{\vec{\mu}}\ket{\vec{\eta};\vec{\upsilon};\vec{\tau}}. \label{eq:measurementprobability}
\end{equation}
This completes our mathematical description of passive optical interferometry. 

\subsection{Permutation of photons and modes}
\label{sec:symmetricgroupaction}
In this subsection we employ the symmetric group to formally define the permutation of photons and modes. The action that permutes photons is also used to define a representation of the symmetric group. In the next subsection the rate matrix is shown to carry this representation.

The symmetric group $S_n$ is the group of all $n!$ permutations of $n$ objects. The state of $n$ photons~\eqref{eq:fockstate}, described by the product of $n$ bosonic creation operators, allows for the definition of the actions of the symmetric group $S_n$ on photonic states. Given an element $\sigma\in S_n$, the action $P_\sigma$ on an $n$-photon state is defined as
\begin{align}
  P_\sigma\ket{\vec{\eta};\vec{\upsilon};\vec{\tau}} &\defeq \frac{1}{\sqrt{\mathcal{N}}}\prod_i A_{\upsilon_{\sigma(i)}}^\dagger(\tau_i)\ket{0} \nonumber \\
&= \ket{\vec{\eta};P_\sigma\vec{\upsilon};\vec{\tau}},
\label{eq:permutationaction}
\end{align}
where $P_\sigma\vec{\upsilon}$ is a reordering of the entries of $\vec{\upsilon}$. There are $N=n!/\prod_i\eta_i!$ distinct permutations of $\vec{\upsilon}$. We arbitrarily order and collect these $N$ permutations in the set
\begin{equation}
  \Upsilon = \set{\vec{\bar{\upsilon}}^1,\dots,\vec{\bar{\upsilon}}^N}. \label{eq:photonoccupationsset}
\end{equation}
The elements of the symmetric group transform between the $N$ states~\eqref{eq:photonoccupationsset} with mode-occupation $\vec{\eta}$.

The action~\eqref{eq:permutationaction} naturally leads to an $N$-dimensional representation of $S_n$. For the permutations $\sigma\in S_n$ its representation $\Gamma(\sigma)$ has entries
\begin{equation}
  \Gamma_{ij}(\sigma) =
  \begin{cases}
    1 \quad\text{ if } P_{\sigma}\vec{\bar{\upsilon}}^i = \vec{\bar{\upsilon}}^j, \\
    0 \quad\text{ otherwise.}
  \end{cases} \label{eq:standardrep}
\end{equation}

We also define $Q$, the action of the symmetric group on photonic states that permutes modes. To describe all mode permutations of $m$-mode states we require the symmetric group $S_m$. Given an element $\sigma\in S_m$,
\begin{align}
  Q_\sigma\ket{\vec{\eta};\vec{\upsilon};\vec{\tau}} &\defeq \frac{1}{\sqrt{\mathcal{N}}}\prod_i A_{\sigma\left(\upsilon_i\right)}^\dagger(\tau_i)\ket{0} \nonumber \\
&= \ket{Q_\sigma\vec{\eta};Q_\sigma\vec{\upsilon}},
\label{eq:permutationmodeaction}
\end{align}
where $Q_\sigma\vec{\eta}$ reorders the entries of $\vec{\eta}$ and $Q_\sigma\vec{\upsilon}$ maps $\upsilon_i$ to $\sigma(\upsilon_i)$. 

The actions $P_\sigma$ and $Q_\sigma$ were previously mentioned when discussing the inputs states shown in Fig.~\ref{fig:examples}. We have the relationship
\begin{align*}
  P_{(23)}\ket{211;1123;\vec{\tau}} &= \ket{211;P_{(23)}1123;\vec{\tau}} \nonumber \\ &=\ket{211;1213;\vec{\tau}},
\end{align*}
where $\sigma=(23)\in S_4$ because it is a system of four photons. We also have the relationship
\begin{align*}
  Q_{(12)}\ket{211;1123;\vec{\tau}} &= \ket{Q_{(12)}211;Q_{(12)}1123;\vec{\tau}} \nonumber \\ &= \ket{121;2213;\vec{\tau}},
\end{align*}
where $\sigma'=(12)\in S_3$ because the example in Fig.~\ref{fig:examples} is a system of three spatial modes. 

\subsection{Coincidence rates for arbitrary inputs}
\label{sec:ratederivation}
We now have all the necessary ingredients to prove that coincidence rates in passive interferometry are of the form~\eqref{eq:rateform}. This is the content of the following proposition.
\begin{theorem}
  \label{th:arbitraryrates}
  Given an $m$-mode $n$-photon state $\ket{\vec{\eta},\vec{\upsilon};\vec{\tau}}$ and an interferometer described by the matrix $U$, then for the output event $\vec{\mu}$, the coincidence rate is
  \begin{align}
    C(\vec{\tau}) &\defeq \bra{\vec{\eta},\vec{\upsilon};\vec{\tau}}U^\dagger M_{\vec{\mu}} U \ket{\vec{\eta},\vec{\upsilon};\vec{\tau}} \nonumber\\&= \vec{u}^\dagger R(\vec{\tau})\vec{u}. \label{eq:propositionrate}
  \end{align}
\end{theorem}

\begin{proof}
  The action of the interferometer~\eqref{eq:interferometeraction} on the state~\eqref{eq:fockstate} results in the state
\begin{equation}
  U\ket{\vec{\eta};\vec{\upsilon};\vec{\tau}} = \sum_{\vec{\upsilon'}\in \set{1,\dots,m}^n} \prod_{k=1}^n U_{\upsilon_k \upsilon_k'} \ket{\vec{\eta}';\vec{\upsilon}';\vec{\tau}},
\end{equation}
the superposition of all possible states with $n$ photons. Here we have suppressed the normalization factor from Eq.~\eqref{eq:fockstate} but calculate it shortly. In the sum above, the only states $\ket{\vec{\eta}';\vec{\upsilon}';\vec{\tau}}$ that have a non-zero contribution to the expectation value
\begin{equation}
 \bra{\vec{\eta},\vec{\upsilon};\vec{\tau}}U^\dagger M_{\vec{\mu}} U \ket{\vec{\eta},\vec{\upsilon};\vec{\tau}}
\end{equation}
are those for which $\vec{\eta}'=\vec{\mu}$. The states for which this condition is true are collected in basis
\begin{equation}
  B = \set{\ket{\vec{\mu};P_\sigma\vec{\xi};\vec{\tau}}: \sigma \in S_n}.
\end{equation}
This basis set has $N = n!/\prod_{i=1}^m \mu_i!$ elements, labeled by the $N$ permutations
\begin{equation}
  \Upsilon = \set{\vec{\bar{\xi}}^1,\dots,\vec{\bar{\xi}}^N}
\end{equation}
of $\vec{\xi}$. Therefore, the coincidence rate
\begin{equation}
  C(\vec{\tau}) = \sum_{i,j=1}^N \prod_{k=1}^n U_{\upsilon_k \bar{\xi}^i_k}^* \bra{\vec{\mu};\vec{\bar{\xi}^i};\vec{\tau}}M_{\vec{\mu}} \prod_{k'=1}^n U_{\upsilon_{k'} \bar{\xi}^j_{k'}}\ket{\vec{\mu};\vec{\bar{\xi}^j};\vec{\tau}}.
\end{equation}
Such a sum of products can equivalently be written as the matrix product $C(\vec{\tau}) = \vec{u}^\dagger R(\vec{\tau})\vec{u}$, where the entries of $\vec{u}$ are
  \begin{equation}
    u_k \defeq \prod_{i=1}^n U_{\upsilon_i\bar{\xi}^k_{i}}, \label{eq:interferometervectorentries}
  \end{equation}
and
\begin{equation}
  R_{ij}(\vec{\tau}) = \bra{\vec{\mu};\vec{\bar{\xi}^i};\vec{\tau}}M_{\vec{\mu}}\ket{\vec{\mu};\vec{\bar{\xi}^j};\vec{\tau}}.
\end{equation}
Employing definition~\eqref{eq:measurementoperator} of $M_{\vec{\mu}}$ and definition~\eqref{eq:fockstate} of photonic states,
\begin{equation}
  R_{ij}(\vec{\tau})=\int \mathrm{d}\omega_1\cdots\mathrm{d}\omega_n r_i r_j^\dagger, \label{eq:ratematrixsimplified}
\end{equation}
where
\begin{equation}
  r_k = \bra{0}\prod_{i=1}^n A_{\bar{\xi}^k_i}(\tau_i) \prod_{j=1}^n a_{\xi_j}^\dagger(\omega_j)\ket{0}. \label{eq:unsimplifiedr}
\end{equation}

The elements of $\vec{r}$ can be simplified as follows. First note that by using the definition of the creation operators~\eqref{eq:bigadefinition} an integral over the frequencies is obtained. This integral can be calculated by defining for element $r_k$ the set
\begin{equation}
  C_k \defeq \set{\sigma\in S_n : P_\sigma \vec{\bar{\xi}}^k = \vec{\xi}}. \label{eq:cosetdef}
\end{equation}
Assuming $\vec{\xi}^1 \defeq \vec{\xi}$, the set $C_1$ stabilizes $\vec{\xi}$ and is therefore a subgroup of $S_n$. The sets $\set{C_k}$ are the right cosets of $S_n$ with respect to $C_1$. Then by using the commutation relations~\eqref{eq:commutationrelations} the integral over the frequencies reduces to the sum
\begin{equation}
  r_k = \sum_{\sigma\in C_k} \phi(\omega_1)\cdots\phi(\omega_n) \e^{\iota\omega_1\tau_{\sigma(1)}}\cdots \e^{\iota\omega_n\tau_{\sigma(n)}}.
\end{equation}
Using this the entries of the rate matrix turn into
\begin{widetext}
\begin{align}
  R_{ij}(\vec{\tau}) = \frac{1}{\mathcal{N}}\frac{1}{\mu_1!,\cdots,\mu_n!}\int & \mathrm{d}\omega_1\cdots\mathrm{d}\omega_n |\phi(\omega_1)|^2\cdots |\phi(\omega_n)|^2\nonumber \\ &\times\left(\sum_{\sigma'\in C_i} \e^{\iota\omega_1\tau_{\sigma'(1)}}\cdots \e^{\iota\omega_n\tau_{\sigma'(n)}}  \right)\left(\sum_{\tilde{\sigma}\in C_j} \e^{-\iota\omega_1\tau_{\tilde{\sigma}(1)}}\cdots \e^{-\iota\omega_n\tau_{\tilde{\sigma}(n)}}\right), \label{eq:ratematrixentries}
\end{align}
where we have reintroduced the normalization factors from Eq.~\eqref{eq:fockstate} and Eq.~\eqref{eq:measurementoperator}. The normalization factor $\mathcal{N}$ is 
\begin{equation}
  \mathcal{N} = \left|\langle \vec{\eta},\vec{\upsilon},\vec{\tau} | \vec{\eta},\vec{\upsilon},\vec{\tau} \rangle\right|^2 = \left|\bra{0}\prod_{i=1}^n A_{\upsilon_i}(\tau_i) \prod_{j=1}^n A_{\upsilon_j}^\dagger(\tau_j)\ket{0}\right|^2.
\end{equation}
This expression is similar to Eq.~\eqref{eq:unsimplifiedr} and can be simplified in the same way. Let $C'$ be subset of elements of $S_n$ that stabilize $\vec{\upsilon}$. Then,
\begin{equation}
  \mathcal{N} = \left|\int \mathrm{d}\omega_1\cdots\mathrm{d}\omega_n |\phi(\omega_1)|^2\cdots |\phi(\omega_n)|^2 \e^{\iota\omega_1\tau_1}\cdots \e^{\iota\omega_n\tau_n} \sum_{\sigma \in C'} \left[\e^{-\iota\omega_1\tau_{\sigma(1)}}\cdots \e^{-\iota\omega_n\tau_{\sigma(n)}}\right]\right|^2. 
\end{equation}
Hence, the coincidence rate~\eqref{eq:propositionrate} can be calculated from Eqs.~\eqref{eq:interferometervectorentries} and~\eqref{eq:ratematrixentries}. \qedhere
\end{widetext}
\end{proof}
Proposition~\ref{th:arbitraryrates} provides convenient expressions to calculate the coincidence rate. The entries of the interferometer vector can be calculated using Eq.~\eqref{eq:interferometervectorentries}, the entries of rate matrix can be calculated using Eq.~\eqref{eq:ratematrixentries}, and the total coincidence rate can then be calculated using Eq.~\eqref{eq:rateform}. The expressions in Proposition~\ref{th:arbitraryrates} are amenable to execution by computer code. \textsc{Mathematica}$^{\circledR}$ code for these expressions, which was also used to do all calculations in this work, can be found on GitHub~\cite{sakhtekiran}. Finally, note that the rate expression that is calculated is not normalized. This is not a problem for predicting the results of photonic interferometry experiments because, due to photon losses, only the relative rates between different outcomes is known and not the absolute value of any given outcome~\cite{Spring2013,Broome2013,Tillmann2013,Spagnolo2014,Bentivegna2015,Tillmann2015}.

The factorization of the rate into a matrix product follows from the linearity of the interferometer's action~\eqref{eq:interferometeraction} on photonic states. The dependence of the entries of the rate matrix on permutations arises due to swapping present in the commutation relations~\eqref{eq:commutationrelations}. The specific form of the wavefunction that we have assumed in Eq.~\eqref{eq:fockstate} is not important, as the same factorization of the rate happens for photons distinguished by (one or more) other non-path degrees of freedom. 

The decomposition of the rate as a matrix product~\eqref{eq:rateform} provides computational shortcuts for calculating coincidence rates for repeated simulations of experiments in two cases. If the interferometer parameters are fixed but the photon delays vary across simulations, then the immanants do not change and have to be computed only once whereas the coefficients do change between computational runs. On the other hand, if the photon delays are fixed but the interferometer parameters vary across simulations, then only the immanants have to be computed anew in each computational run.

\subsection{Coincidence rates in terms of immanants}
\label{sec:ratesimmanants}
In this section, we show that the coincidence rates can be expressed in terms of immanants and provide a procedure to do so. First, we show that the rate matrix carries a representation of the symmetric group, and hence can be block-diagonalized using standard methods of the symmetric group representation theory~\cite{Chen1989}. Second, we provide a procedure to convert the entries of the rotated interferometer vector into sums of immanants. Then, Eq.~\eqref{eq:rotatedrate} can be used to calculate the coincidence rate.

We prove that the rate matrix~\eqref{eq:ratematrixentries} can be block-diagonalized with the aid of the following proposition.

~
\begin{theorem}
  \label{th:ratematrixrep}
  The rate matrix carries the representation~\eqref{eq:standardrep}; i.e.\ the rate matrix may be expressed as
\begin{equation}
    R(\vec{\tau}) = \sum_{\sigma\in S_n} \Delta_\sigma(\vec{\tau})\Gamma(\sigma). \label{eq:ratematrixsum}
  \end{equation}
\end{theorem}

\emph{Proof:} First note that Eqs.~\eqref{eq:ratematrixsum} and~\eqref{eq:standardrep} imply that
\begin{equation}
  R_{ij}(\vec{\tau}) = \sum_{\sigma \in S_n} 
  \begin{cases}
    \Delta_\sigma(\vec{\tau}) &\text{ if } P_{\sigma}\vec{\bar{\xi}}^i = \vec{\bar{\xi}}^j \\
    0 &\text{ otherwise}
  \end{cases}. \label{eq:ratematrixrepresentationbreakdown}
\end{equation}
We show that Eq.~\eqref{eq:ratematrixentries} is equivalent to this expression. The double sum in Eq.~\eqref{eq:ratematrixentries} has entries
\begin{align}
  \int &\mathrm{d}\omega_1\cdots\mathrm{d}\omega_n |\phi(\omega_1)|^2\cdots |\phi(\omega_n)|^2 \nonumber \\&\times \e^{\iota\omega_1(\tau_{\sigma'(1)}-\tau_{\tilde{\sigma}(1)})}\cdots \e^{\iota\omega_n(\tau_{\sigma'(n)}-\tau_{\tilde{\sigma}(n)})} \\
  = \int &\mathrm{d}\omega_1\cdots\mathrm{d}\omega_n |\phi(\omega_1)|^2\cdots |\phi(\omega_n)|^2 \nonumber \\&\times \e^{\iota\omega_1(\tau_{\tilde{\sigma}^{-1}\sigma'(1)}-\tau_{1})}\cdots \e^{\iota\omega_n(\tau_{\tilde{\sigma}^{-1}\sigma'(n)}-\tau_{n})}, \label{eq:doublesumentries}
\end{align}
where the equality is obtained by a suitable relabeling of the $\set{\omega_i}$. Note that from the definitions of $\sigma' \in C_i$ and $\tilde{\sigma} \in C_j$ it follows that $(P_{\tilde{\sigma}})^{-1}P_{\sigma'}\vec{\bar{\xi}}^i = \vec{\bar{\xi}}^j$. Hence, if we set $\sigma = \tilde{\sigma}^{-1}\sigma'$ in Eq.~\eqref{eq:doublesumentries} then Eq.~\eqref{eq:ratematrixrepresentationbreakdown} is satisfied. Explicitly, by comparing Eq.~\eqref{eq:ratematrixentries} and Eq.~\eqref{eq:ratematrixrepresentationbreakdown}, it is clear that 
\begin{widetext}
\begin{equation}
  \Delta_{\tilde{\sigma}^{-1}\sigma'}(\vec{\tau}) = \int \mathrm{d}\omega_1\cdots\mathrm{d}\omega_n |\phi(\omega_1)|^2\cdots |\phi(\omega_n)|^2\left( \e^{\iota\omega_1\tau_{\sigma'(1)}}\cdots \e^{\iota\omega_n\tau_{\sigma'(n)}}  \right)\left( \e^{-\iota\omega_1\tau_{\tilde{\sigma}(1)}}\cdots \e^{-\iota\omega_n\tau_{\tilde{\sigma}(n)}}\right). \qed
\end{equation}
\end{widetext}

To block-diagonalize the rate matrix requires the calculations of the basis $V$ referred to in Eq.~\eqref{eq:rotatedrate}. This basis can be obtained using a variety of methods, one of which is provided in Chap.~$4$ of Ref.~\cite{Chen1989}. An example application of this method is presented in Appendix~\ref{sec:appendixexample} and an implementation can be found in our reference code~\cite{sakhtekiran}.

Now, we provide a procedure to express the entries of the rotated interferometer vector $V\vec{u}$ as sums of immanants of the scattering matrix. As discussed in Sec.~\ref{sec:mainexamplerate}, the standard representation $\Gamma$ and the space $H^{\vec{\mu}}$ can be block-diagonalized. Young tableaux, obtained by filling Young diagrams with positive integers, provide information about the size and multiplicity of the blocks that occur in these block-diagonalization. The standard representation $\Gamma$ block-diagonalizes as
\begin{equation}
  \Gamma = \bigoplus_{\lambda \vdash n} p^\lambda \lambda,
\end{equation}
where the sum is over all partitions of $n$ and $p^\lambda$ is the multiplicity of the irrep $\lambda$. A standard Young tableau is an $n$-box Young diagram filled with $\{1,\dots,n\}$ such that numbers strictly increase across rows and columns. The dimension $d$ of irrep $\lambda$ is the number of standard Young tableaux of $\lambda$. A semi-standard Young tableau is an $n$-box Young diagram filled by $n$, possibly repeating, positive integers such that the numbers strictly increase across columns and weakly across rows. The multiplicity $p^\lambda$ is the number of semi-standard Young tableaux formed using Young diagram $\lambda$ and the string $\vec{\xi}$. This information about the decomposition of $\Gamma$ can be useful for determining the size and number of blocks in the decomposition of the rate matrix~\eqref{eq:ratematrixsum}.

The space $H^{\vec{\mu}}$ block-diagonalizes as
\begin{equation}
  H^{\vec{\mu}} = \bigoplus_{\lambda \vdash n} p^\lambda H^\lambda, \label{eq:subspacedecomposition}
\end{equation}
where the sum is over all partitions of $n$ and $p^\lambda$ is the multiplicity of the the subspace $H^\lambda$, which is determined as above. The subspace $H^\lambda$ is spanned by up to $d$ distinct immanants of $\lambda$ type, which will happen if there is no multiplicity of photons in the input or output modes. If there are multiple photons in any of the input or output modes, then the number of distinct immanants is reduced, because the number of distinct states that can be obtained by permuting the photons is reduced. The rule for determining the actual number of immanants of certain type $\lambda$ is rather simple. Count the number of semi-standard Young tableaux of $\lambda$, once using input photon-occupation string $\vec{\upsilon}$ and once using the output photon-occupation string $\vec{\mu}$: the minimum of these two counts is the number of distinct immanants of the type $\lambda$. Let this set of distinct immanants be $I^\lambda$. 

Now, we discuss how to express the entries of the interferometer vector in terms of immanants. Recall, that $H^{\vec{\mu}}$ reduces in the block-diagonal basis $V$. Therefore, $V\vec{u}$ is also reduced and has blocks corresponding to the irreps of $S_n$. This means that the entries of $V\vec{u}$ can be linearly transformed to express them in terms of immanants. Explicitly, suppose the element $(V\vec{u})_i$ is associated with the irrep $\lambda$. Then, the element $(V\vec{u})_i$ is a vector in the space $H^\lambda$, and can be expressed in terms of the elements of the set $I^\lambda$. In this way, $V\vec{u}$ and the coincidence rate~\eqref{eq:rotatedrate} can be expressed in terms of immanants.

We finish this section with a short discussion on how to calculate coincidence rates for input states related by the permutation of either photons or modes. The technical results have already been proved in Secs.~\ref{sec:permutedphotons} and \ref{sec:permutedmodes}. If for some input state the rate matrix and interferometer vector have been calculated, then for any other input related to the first by a permutation of photons the coincidence rate can be calculated using Eq.~\eqref{eq:permutedphotons}. If the new input is related by a permutation of modes, the the rate matrix remains the same but the interferometer vector has to be calculated anew with the new scattering matrix. This completes our discussion on how to calculate coincidence rates for any input and output configuration of photons.

\section{Conclusion}
\label{ref:conclusion}
We have developed a theory of passive optical interferometry that relates the coincidence rates at the output of an interferometer with the permutational symmetries of the input photons. The permutational symmetry of the input photons can be controlled by varying their distinguishability, which in our case is done by tuning the time-of-arrival of the photons at the interferometer. Our results are obtained by exploiting the representation theory of the symmetric group.

The coincidence rates at the output of the interferometer are expressed in terms of the immanants of the scattering matrix, where the relative weights of the immanants depends on the distinguishability of the input photons. If any exchange symmetries of the input photons are forbidden because of the multiplicity of photons in input ports then immanants with those exchange symmetries do not appear in the coincidence rate expressions. For inputs related by a permutation their rates are correspondingly related by a simple linear transformation.

As a simple example of the applicability of our formalism we presented a study of the HOM experiment with photon sources that sometimes output two photons. Though for this system direct computation of the rates is straightforward, our theory can be used to infer without any computation the qualitative changes to the coincidence rates due to multiple input photons. Such qualitative analysis can also be applied to larger interference experiments.

\section*{Acknowledgments}
We thank M. Broome for helpful discussions and M. Frood for his contribution to the early stages of this manuscript. B.C.S. acknowledges financial support from Alberta Innovates and NSERC. H.de G. and D.J.S. acknowledge financial support
from NSERC and Lakehead University. A.K. acknowledges financial support from NSERC CREATE CryptoWorks21.

\bibliographystyle{unsrt}
\bibliography{photonicsymmetries}

\appendix
\begin{widetext}
\allowdisplaybreaks
\section{Four photons in two modes}
\label{sec:appendixexample}
In this appendix we fully work out the details of the examples presented in Sec.~\ref{sec:example}. We provide explicit expressions for the interferometer vector and the rate matrix. We use class operator methods to find the basis $V$ in which the rate matrix is block-diagonal and the entries of the interferometer matrix are the immanants of the scattering matrix.

The input state,
\begin{equation}
  \ket{211;1123} = A_1^\dagger(\tau_1)A_1^\dagger(\tau_2)A_2^\dagger(\tau_3)A_3^\dagger(\tau_4)\ket{0}, 
\end{equation}
transforms under the action of the interferometer to
\begin{align}
  U\ket{211;1123} =& \left(U_{11}A_1^\dagger(\tau_1)+U_{12}A_2^\dagger(\tau_1)+U_{13}A_3^\dagger(\tau_1)\right)\left(U_{11}A_1^\dagger(\tau_2)+U_{12}A_2^\dagger(\tau_2)+U_{13}A_3^\dagger(\tau_2)\right)\nonumber \\&
\times\left(U_{21}A_1^\dagger(\tau_3)+U_{22}A_2^\dagger(\tau_3)+U_{23}A_3^\dagger(\tau_3)\right)\left(U_{31}A_1^\dagger(\tau_4)+U_{32}A_2^\dagger(\tau_4)+U_{33}A_3^\dagger(\tau_4)\right)\ket{0}. \label{eq:scattered211}
\end{align}
We have ignored the normalization factor in these expressions for simplicity. We are interested in the output event $\vec{\mu}=022$. Of the $81$ terms in the expansion of~\eqref{eq:scattered211}, only six correspond to two photons in each of the second and third modes, coinciding with the six permutations of $\vec{\xi}=2233$. Only these six terms give a non-zero contribution to the the rate~\eqref{eq:measurementprobability}. The amplitude of these six terms is $w^\dagger\ket{0}\vec{u}$, where the interferometer vector
\begin{equation}
  \vec{u} =
  \begin{pmatrix}
    U_{12}U_{12}U_{23}U_{33} \\
    U_{12}U_{13}U_{22}U_{33} \\
    U_{12}U_{13}U_{23}U_{32} \\
    U_{13}U_{12}U_{22}U_{33} \\
    U_{13}U_{12}U_{23}U_{32} \\
    U_{13}U_{13}U_{22}U_{32} \\
  \end{pmatrix}, \label{eq:interferometervectorexample}
\end{equation}
and vector
\begin{equation}
  \vec{w} = \begin{pmatrix}
A_2(\tau_1)A_2(\tau_2)A_3(\tau_3)A_3(\tau_4) \\
A_2(\tau_1)A_3(\tau_2)A_2(\tau_3)A_3(\tau_4)\\
A_2(\tau_1)A_3(\tau_2)A_3(\tau_3)A_2(\tau_4)\\
A_3(\tau_1)A_2(\tau_2)A_2(\tau_3)A_3(\tau_4)\\
A_3(\tau_1)A_2(\tau_2)A_3(\tau_3)A_2(\tau_4)\\
A_3(\tau_1)A_3(\tau_2)A_2(\tau_3)A_2(\tau_4)\\
  \end{pmatrix}.
\end{equation}
Employing the measurement operator $M_{022}$ for detecting two photons in each mode the rate
\begin{align}
  C(\vec{\tau}) &= \bra{221;1123}U^\dagger M_{022}U\ket{221;1123} \nonumber\\
  &= \vec{u}^\dagger\bra{0}\vec{w}M_{022} \vec{w}^\dagger\ket{0} \vec{u} 
\end{align}
The $6\times 6$ matrix
\begin{equation}
  R(\vec{\tau}) \defeq \bra{0}\vec{w} M_{022} \vec{w}^\dagger \ket{0} \label{eq:ratematrix22}
\end{equation}
is the rate matrix. This expression is equivalent to the intermediate expression~\eqref{eq:ratematrixsimplified} of the rate matrix in Proposition~\ref{th:arbitraryrates}. The rate matrix can be calculated from this expression using tedious calculations. We use the simplified expression~\eqref{eq:ratematrixentries} to calculate the rate matrix. First we note that the string $\vec{\xi}=2233$ has permutations
\begin{equation}
  \Upsilon = \set{2233,2323,2332,3223,3232,3322}. \label{eq:exampleupsilon}
\end{equation}
Using this, we can find the cosets of $\vec{\xi}=2233$. The string $\vec{\xi}=2233$ is stabilized by the $S_4$ subgroup $C_1=\{\mathrm{e},(12),(34),(12,34)\}$: for instance $P_{12,34}2233=2233$. This subgroup has right cosets
\begin{align}
  C_1 &= \{\mathrm{e},(12),(34),(12,34)\}, \nonumber\\
  C_2 &= \{(23),(132),(234),(1342)\}, \nonumber\\
  C_3 &= \{(24),(142),(243),(1432)\}, \nonumber\\
  C_4 &= \{(13),(123),(134),(1234)\}, \nonumber\\
  C_5 &= \{(14),(124),(143),(1243)\}, \nonumber\\
  C_6 &= \{(1324),(1423),(13,24),(14,23)\}.
\end{align}
By definition~\eqref{eq:cosetdef}, the elements of coset $C_k$ maps $\vec{\bar{\xi}}^k$ to $\vec{\xi}=\vec{\bar{\xi}}^1$: for instance $(23)\in C_2$ acts as $P_{23}2323=2233$. 

To calculate an explicit expression for the rate matrix we assume that all photons have an identical Gaussian spectrum, 
\begin{equation}
  \phi_i(\omega) = \frac{1}{\left(\pi\sigma^2\right)^{1/4}}\exp\left(-\frac{(\omega-\omega_0)^2}{2\sigma^2}\right) \quad\forall i\in\set{1,\dots,n}.
\end{equation}
For modeling a real experiment, the rates can be calculated using experimentally measured spectrums of the input photons. Using this Gaussian spectrum and employing the expression~\eqref{eq:ratematrixentries} we find that the rate matrix is
\begin{equation}
  R(\vec{\tau}) = 
  \begin{pmatrix}
    \tilde R_1 & \tilde R_2 & \tilde R_3
  \end{pmatrix},
\end{equation}
where
\begin{align}
  \tilde{R}_1 &=
  \begin{pmatrix}
    \Delta_{\e}+\Delta_{(12)}+\Delta_{(12,34)}+\Delta_{(34)}          & \Delta_{(132)}+\Delta_{(1342)}+\Delta_{(23)}+\Delta_{(234)}           \\
    \Delta_{(123)}+\Delta_{(1243)}+\Delta_{(23)}+\Delta_{(243)}       & \Delta_{\e}+\Delta_{(13)}+\Delta_{(13,24)}+\Delta_{(24)}             \\
    \Delta_{(1234)}+\Delta_{(124)}+\Delta_{(234)}+\Delta_{(24)}        & \Delta_{(1324)}+\Delta_{(134)}+\Delta_{(243)}+\Delta_{(34)}           \\
    \Delta_{(13)}+\Delta_{(132)}+\Delta_{(143)}+\Delta_{(1432)}        & \Delta_{(12)}+\Delta_{(123)}+\Delta_{(142)}+\Delta_{(1423)}          \\
    \Delta_{(134)}+\Delta_{(1342)}+\Delta_{(14)}+\Delta_{(142)}        & \Delta_{(12,34)}+\Delta_{(1234)}+\Delta_{(14,23)}+\Delta_{(1432)}         \\
    \Delta_{(13,24)}+\Delta_{(1324)}+\Delta_{(14,23)}+\Delta_{(1423)}  & \Delta_{(124)}+\Delta_{(1243)}+\Delta_{(14)}+\Delta_{(143)}       & 
  \end{pmatrix}, \\
  \tilde{R}_2 &=
  \begin{pmatrix}
    \Delta_{(142)}+\Delta_{(1432)}+\Delta_{(24)}+\Delta_{(243)} & \Delta_{(123)}+\Delta_{(1234)}+\Delta_{(13)}+\Delta_{(134)} \\
    \Delta_{(1423)}+\Delta_{(143)}+\Delta_{(234)}+\Delta_{(34)} & \Delta_{(12)}+\Delta_{(124)}+\Delta_{(132)}+\Delta_{(1324)} \\
    \Delta_{\e}+\Delta_{(14)}+\Delta_{(14,23)}+\Delta_{(23)} &  \Delta_{(12,34)}+\Delta_{(1243)}+\Delta_{(13,24)}+\Delta_{(1342)} \\
    \Delta_{(12,34)}+\Delta_{(1243)}+\Delta_{(13,24)}+\Delta_{(1342)} & \Delta_{\e}+\Delta_{(14)}+\Delta_{(14,23)}+\Delta_{(23)} \\
    \Delta_{(12)}+\Delta_{(124)}+\Delta_{(132)}+\Delta_{(1324)} & \Delta_{(1423)}+\Delta_{(143)}+\Delta_{(234)}+\Delta_{(34)} \\
    \Delta_{(123)}+\Delta_{(1234)}+\Delta_{(13)}+\Delta_{(134)} & \Delta_{(142)}+\Delta_{(1432)}+\Delta_{(24)}+\Delta_{(243)}
  \end{pmatrix} \\
  \tilde{R}_3 &= \begin{pmatrix}
    \Delta_{(124)}+\Delta_{(1243)}+\Delta_{(14)}+\Delta_{(143)}       & \Delta_{(13,24)}+\Delta_{(1324)}+\Delta_{(14,23)}+\Delta_{(1423)} \\
    \Delta_{(12,34)}+\Delta_{(1234)}+\Delta_{(14,23)}+\Delta_{(1432)} & \Delta_{(134)}+\Delta_{(1342)}+\Delta_{(14)}+\Delta_{(142)}       \\
    \Delta_{(12)}+\Delta_{(123)}+\Delta_{(142)}+\Delta_{(1423)}       & \Delta_{(13)}+\Delta_{(132)}+\Delta_{(143)}+\Delta_{(1432)}       \\
    \Delta_{(1324)}+\Delta_{(134)}+\Delta_{(243)}+\Delta_{(34)}       & \Delta_{(1234)}+\Delta_{(124)}+\Delta_{(234)}+\Delta_{(24)}       \\
    \Delta_{\e}+\Delta_{(13)}+\Delta_{(13,24)}+\Delta_{(24)}   & \Delta_{(123)}+\Delta_{(1243)}+\Delta_{(23)}+\Delta_{(243)} \\
    \Delta_{(132)}+\Delta_{(1342)}+\Delta_{(23)}+\Delta_{(234)} & \Delta_{\e}+\Delta_{(12)}+\Delta_{(12,34)}+\Delta_{(34)}  \\
\end{pmatrix},                                                                  
\end{align}
and
\begin{align}
  \Delta_{\e} &= 1,
  & \Delta_{(12)} &= \e^{-\sigma^2\left(\tau_1-\tau_2\right)^2}, \nonumber\\
  \Delta_{(13)} &= e^{-\sigma^2\left(\tau_1-\tau_3\right)^2}, 
  & \Delta_{(14)} &= e^{-\sigma^2\left(\tau_1-\tau_4\right)^2}, \nonumber\\
  \Delta_{(23)} &= \e^{-\sigma^2\left(\tau _2-\tau_3\right)^2}, 
  & \Delta_{(24)} &= \e^{-\sigma^2\left(\tau_2-\tau_4\right)^2}, \nonumber\\
  \Delta_{(34)} &= \e^{-\sigma^2\left(\tau_3-\tau_4\right)^2}, 
  & \Delta_{(123)} &= \e^{-\sigma^2\left(\tau_1^2+\tau_2^2+\tau_3^2-\tau_1\tau_2-\tau_1\tau_3-\tau_2 \tau_3\right)}, \nonumber\\
  \Delta_{(124)} &= \e^{-\sigma^2\left(\tau_1^2+\tau_2^2+\tau_4^2-\tau_1\tau_2-\tau_1\tau_4-\tau_2\tau_4\right)}, 
  & \Delta_{(132)} &= \Delta_{123} \nonumber\\
  \Delta_{(134)} &= \e^{-\sigma^2\left(\tau_1^2+\tau_3^2+\tau_4^2-\tau_1\tau_3-\tau_1\tau_4-\tau_3\tau_4\right)}, 
  & \Delta_{(142)} &= \Delta_{124} \nonumber\\
  \Delta_{(143)} &= \Delta_{134} 
  & \Delta_{(234)} &= \e^{-\sigma^2\left(\tau_2^2+\tau_3^2+\tau_4^2-\tau_2\tau_3-\tau_2\tau_4-\tau_3\tau_4\right)}, \nonumber\\
  \Delta_{(243)} &= \Delta_{234} 
  & \Delta_{(1234)} &= \e^{-\sigma^2\left(\tau_1^2+\tau_2^2+\tau_3^2+\tau_4^2-\tau_1\tau_2-\tau_1\tau_4 -\tau_2 \tau_3-\tau_3 \tau_4\right)}, \nonumber\\
  \Delta_{(1243)} &= \e^{-\sigma^2\left(\tau_1^2+\tau_2^2+\tau_3^2+\tau_4^2-\tau_1\tau_2-\tau_1\tau_3-\tau_2\tau_4-\tau_3\tau_4\right)}, 
  & \Delta_{(1324)} &= \e^{-\sigma^2\left(\tau_1^2+\tau_2^2+\tau_3^2+\tau_4^2-\tau_1\tau_3-\tau_1\tau_4-\tau_2\tau_3-\tau_2\tau_4\right)}, \nonumber\\
  \Delta_{(1342)} &= \Delta_{1243}, 
  & \Delta_{(1423)} &= \Delta_{1324}, \nonumber\\
  \Delta_{(1432)} &= \Delta_{1234}, 
  & \Delta_{(12,34)} &= \e^{-\sigma^2\left[\left(\tau _1-\tau_2\right)^2+\left(\tau_3-\tau_4\right)^2\right]}, \nonumber\\
  \Delta_{(13,24)} &= \e^{-\sigma^2\left[\left(\tau_1-\tau_3\right)^2+\left(\tau_2-\tau_4\right)^2\right]}, 
  & \Delta_{(14,23)} &= \e^{-\sigma^2\left[\left(\tau_1-\tau_4\right)^2+\left(\tau_2-\tau_3\right)^2\right]}.
\end{align}
This compact form of the rate matrix in terms of $\Delta_\sigma$ is provided because we also want to verify that the rate matrix carries a representation of the symmetric group. Using Eq.~\eqref{eq:exampleupsilon} we find that the representation $\Gamma$ has basis
\begin{equation}
  B = \set{\ket{022;2233;\vec{\tau}},\ket{022;2323;\vec{\tau}},\ket{022;2332;\vec{\tau}},\ket{022;3223;\vec{\tau}},\ket{022;3232;\vec{\tau}},\ket{022;3322;\vec{\tau}}}.
\end{equation}
In this basis, using Eq.~\eqref{eq:standardrep}, the representation of the elements of $S_4$ are
\begin{align}
  \Gamma_{\e} &= \begin{pmatrix}
    1 & 0 & 0 & 0 & 0 & 0 \\
    0 & 1 & 0 & 0 & 0 & 0 \\
    0 & 0 & 1 & 0 & 0 & 0 \\
    0 & 0 & 0 & 1 & 0 & 0 \\
    0 & 0 & 0 & 0 & 1 & 0 \\
    0 & 0 & 0 & 0 & 0 & 1 \\
  \end{pmatrix}, \quad \Gamma_{(12)} =
  \begin{pmatrix}
    1 & 0 & 0 & 0 & 0 & 0 \\
    0 & 0 & 0 & 1 & 0 & 0 \\
    0 & 0 & 0 & 0 & 1 & 0 \\
    0 & 1 & 0 & 0 & 0 & 0 \\
    0 & 0 & 1 & 0 & 0 & 0 \\
    0 & 0 & 0 & 0 & 0 & 1
  \end{pmatrix}, \quad \Gamma_{(13)} =
  \begin{pmatrix}
    0 & 0 & 0 & 1 & 0 & 0 \\
    0 & 1 & 0 & 0 & 0 & 0 \\
    0 & 0 & 0 & 0 & 0 & 1 \\
    1 & 0 & 0 & 0 & 0 & 0 \\
    0 & 0 & 0 & 0 & 1 & 0 \\
    0 & 0 & 1 & 0 & 0 & 0
  \end{pmatrix}, \nonumber \\
  \Gamma_{(14)} &= \begin{pmatrix}
    0 & 0 & 0 & 0 & 1 & 0 \\
    0 & 0 & 0 & 0 & 0 & 1 \\
    0 & 0 & 1 & 0 & 0 & 0 \\
    0 & 0 & 0 & 1 & 0 & 0 \\
    1 & 0 & 0 & 0 & 0 & 0 \\
    0 & 1 & 0 & 0 & 0 & 0
  \end{pmatrix}, \quad \Gamma_{(23)} =
  \begin{pmatrix}
    0 & 1 & 0 & 0 & 0 & 0 \\
    1 & 0 & 0 & 0 & 0 & 0 \\
    0 & 0 & 1 & 0 & 0 & 0 \\
    0 & 0 & 0 & 1 & 0 & 0 \\
    0 & 0 & 0 & 0 & 0 & 1 \\
    0 & 0 & 0 & 0 & 1 & 0
  \end{pmatrix}, \quad \Gamma_{(24)} =
  \begin{pmatrix}
    0 & 0 & 1 & 0 & 0 & 0 \\
    0 & 1 & 0 & 0 & 0 & 0 \\
    1 & 0 & 0 & 0 & 0 & 0 \\
    0 & 0 & 0 & 0 & 0 & 1 \\
    0 & 0 & 0 & 0 & 1 & 0 \\
    0 & 0 & 0 & 1 & 0 & 0
  \end{pmatrix}, \nonumber \\
  \Gamma_{(34)} &= \begin{pmatrix}
    1 & 0 & 0 & 0 & 0 & 0 \\
    0 & 0 & 1 & 0 & 0 & 0 \\
    0 & 1 & 0 & 0 & 0 & 0 \\
    0 & 0 & 0 & 0 & 1 & 0 \\
    0 & 0 & 0 & 1 & 0 & 0 \\
    0 & 0 & 0 & 0 & 0 & 1
  \end{pmatrix}, \quad \Gamma_{(123)} =
  \begin{pmatrix}
    0 & 0 & 0 & 1 & 0 & 0 \\
    1 & 0 & 0 & 0 & 0 & 0 \\
    0 & 0 & 0 & 0 & 1 & 0 \\
    0 & 1 & 0 & 0 & 0 & 0 \\
    0 & 0 & 0 & 0 & 0 & 1 \\
    0 & 0 & 1 & 0 & 0 & 0 \\
  \end{pmatrix}, \quad \Gamma_{(124)} =
  \begin{pmatrix}
    0 & 0 & 0 & 0 & 1 & 0 \\
    0 & 0 & 0 & 1 & 0 & 0 \\
    1 & 0 & 0 & 0 & 0 & 0 \\
    0 & 0 & 0 & 0 & 0 & 1 \\
    0 & 0 & 1 & 0 & 0 & 0 \\
    0 & 1 & 0 & 0 & 0 & 0 \\
  \end{pmatrix}, \nonumber \\
  \Gamma_{(132)} &= \begin{pmatrix}
    0 & 1 & 0 & 0 & 0 & 0 \\
    0 & 0 & 0 & 1 & 0 & 0 \\
    0 & 0 & 0 & 0 & 0 & 1 \\
    1 & 0 & 0 & 0 & 0 & 0 \\
    0 & 0 & 1 & 0 & 0 & 0 \\
    0 & 0 & 0 & 0 & 1 & 0 \\
  \end{pmatrix}, \quad \Gamma_{(134)} =
  \begin{pmatrix}
    0 & 0 & 0 & 1 & 0 & 0 \\
    0 & 0 & 0 & 0 & 0 & 1 \\
    0 & 1 & 0 & 0 & 0 & 0 \\
    0 & 0 & 0 & 0 & 1 & 0 \\
    1 & 0 & 0 & 0 & 0 & 0 \\
    0 & 0 & 1 & 0 & 0 & 0 \\
  \end{pmatrix}, \quad \Gamma_{(142)} =
  \begin{pmatrix}
    0 & 0 & 1 & 0 & 0 & 0 \\
    0 & 0 & 0 & 0 & 0 & 1 \\
    0 & 0 & 0 & 0 & 1 & 0 \\
    0 & 1 & 0 & 0 & 0 & 0 \\
    1 & 0 & 0 & 0 & 0 & 0 \\
    0 & 0 & 0 & 1 & 0 & 0 \\
  \end{pmatrix}, \nonumber \\
  \Gamma_{(143)} &= \begin{pmatrix}
    0 & 0 & 0 & 0 & 1 & 0 \\
    0 & 0 & 1 & 0 & 0 & 0 \\
    0 & 0 & 0 & 0 & 0 & 1 \\
    1 & 0 & 0 & 0 & 0 & 0 \\
    0 & 0 & 0 & 1 & 0 & 0 \\
    0 & 1 & 0 & 0 & 0 & 0 \\
  \end{pmatrix}, \quad \Gamma_{(234)} =
  \begin{pmatrix}
    0 & 1 & 0 & 0 & 0 & 0 \\
    0 & 0 & 1 & 0 & 0 & 0 \\
    1 & 0 & 0 & 0 & 0 & 0 \\
    0 & 0 & 0 & 0 & 0 & 1 \\
    0 & 0 & 0 & 1 & 0 & 0 \\
    0 & 0 & 0 & 0 & 1 & 0 \\
  \end{pmatrix} \quad \Gamma_{(243)} =
  \begin{pmatrix}
    0 & 0 & 1 & 0 & 0 & 0 \\
    1 & 0 & 0 & 0 & 0 & 0 \\
    0 & 1 & 0 & 0 & 0 & 0 \\
    0 & 0 & 0 & 0 & 1 & 0 \\
    0 & 0 & 0 & 0 & 0 & 1 \\
    0 & 0 & 0 & 1 & 0 & 0 \\
  \end{pmatrix}, \nonumber \\
  \Gamma_{(1234)} &= \begin{pmatrix}
    0 & 0 & 0 & 1 & 0 & 0 \\
    0 & 0 & 0 & 0 & 1 & 0 \\
    1 & 0 & 0 & 0 & 0 & 0 \\
    0 & 0 & 0 & 0 & 0 & 1 \\
    0 & 1 & 0 & 0 & 0 & 0 \\
    0 & 0 & 1 & 0 & 0 & 0 \\
  \end{pmatrix}, \quad \Gamma_{(1243)} =
  \begin{pmatrix}
    0 & 0 & 0 & 0 & 1 & 0 \\
    1 & 0 & 0 & 0 & 0 & 0 \\
    0 & 0 & 0 & 1 & 0 & 0 \\
    0 & 0 & 1 & 0 & 0 & 0 \\
    0 & 0 & 0 & 0 & 0 & 1 \\
    0 & 1 & 0 & 0 & 0 & 0 \\
  \end{pmatrix}, \quad \Gamma_{(1324)} =
  \begin{pmatrix}
    0 & 0 & 0 & 0 & 0 & 1 \\
    0 & 0 & 0 & 1 & 0 & 0 \\
    0 & 1 & 0 & 0 & 0 & 0 \\
    0 & 0 & 0 & 0 & 1 & 0 \\
    0 & 0 & 1 & 0 & 0 & 0 \\
    1 & 0 & 0 & 0 & 0 & 0 \\
  \end{pmatrix}, \nonumber \\
  \Gamma_{(1342)} &= \begin{pmatrix}
    0 & 1 & 0 & 0 & 0 & 0 \\
    0 & 0 & 0 & 0 & 0 & 1 \\
    0 & 0 & 0 & 1 & 0 & 0 \\
    0 & 0 & 1 & 0 & 0 & 0 \\
    1 & 0 & 0 & 0 & 0 & 0 \\
    0 & 0 & 0 & 0 & 1 & 0 \\
  \end{pmatrix}, \quad \Gamma_{(1423)} =
  \begin{pmatrix}
    0 & 0 & 0 & 0 & 0 & 1 \\
    0 & 0 & 1 & 0 & 0 & 0 \\
    0 & 0 & 0 & 0 & 1 & 0 \\
    0 & 1 & 0 & 0 & 0 & 0 \\
    0 & 0 & 0 & 1 & 0 & 0 \\
    1 & 0 & 0 & 0 & 0 & 0 \\
  \end{pmatrix}, \quad \Gamma_{(1432)} =
  \begin{pmatrix}
    0 & 0 & 1 & 0 & 0 & 0 \\
    0 & 0 & 0 & 0 & 1 & 0 \\
    0 & 0 & 0 & 0 & 0 & 1 \\
    1 & 0 & 0 & 0 & 0 & 0 \\
    0 & 1 & 0 & 0 & 0 & 0 \\
    0 & 0 & 0 & 1 & 0 & 0 \\
  \end{pmatrix}, \nonumber \\
  \Gamma_{(12,34)} &= \begin{pmatrix}
    1 & 0 & 0 & 0 & 0 & 0 \\
    0 & 0 & 0 & 0 & 1 & 0 \\
    0 & 0 & 0 & 1 & 0 & 0 \\
    0 & 0 & 1 & 0 & 0 & 0 \\
    0 & 1 & 0 & 0 & 0 & 0 \\
    0 & 0 & 0 & 0 & 0 & 1 \\
  \end{pmatrix}, \quad \Gamma_{(13,24)} =
  \begin{pmatrix}
    0 & 0 & 0 & 0 & 0 & 1 \\
    0 & 1 & 0 & 0 & 0 & 0 \\
    0 & 0 & 0 & 1 & 0 & 0 \\
    0 & 0 & 1 & 0 & 0 & 0 \\
    0 & 0 & 0 & 0 & 1 & 0 \\
    1 & 0 & 0 & 0 & 0 & 0 \\
  \end{pmatrix}, \quad \Gamma_{(14,23)} =
  \begin{pmatrix}
    0 & 0 & 0 & 0 & 0 & 1 \\
    0 & 0 & 0 & 0 & 1 & 0 \\
    0 & 0 & 1 & 0 & 0 & 0 \\
    0 & 0 & 0 & 1 & 0 & 0 \\
    0 & 1 & 0 & 0 & 0 & 0 \\
    1 & 0 & 0 & 0 & 0 & 0 \\
  \end{pmatrix}.
\end{align}
Using these representations it is easy to verify via inspection that $R(\vec{\tau}) = \sum_{\sigma\in S_4} \Delta_\sigma(\vec{\tau})\Gamma_\sigma$, which means the rate matrix does carry the representation $\Gamma$ of $S_4$. Therefore, in order to block-diagonalize the rate matrix we have to reduce $\Gamma$ to the irreps of $S_4$.

First we discuss which irreps do occur in this reduction. The irreps of $S_n$ are labeled by $n$-box Young diagrams. The irreps of $S_4$ are labeled by the Young diagrams $\inlineyng{4}$, $\inlineyng{3,1}$, $\inlineyng{2,2}$, $\inlineyng{2,1,1}$ and $\inlineyng{1,1,1,1}$. The rule for determining which irreps do occur is to count the number of semi-standard Young tableaus for each irrep using the string $\vec{\xi}=2233$. For each of the first three Young diagrams there is only one possible semi-standard Young tableau. These are
\begin{equation}
  {\Yboxdim{10pt}\young(2233), \qquad \young(223,3), \qquad \young(22,33)}. \label{eq:semistandardtableaus}
\end{equation}
For the last two Young diagrams, $\inlineyng{2,1,1}$ and $\inlineyng{1,1,1,1}$ there is no possible way to construct a valid semi-standard Young tableau. Therefore, the corresponding irreps do not occur in the decomposition. $\Gamma$ decomposes as
\begin{equation}
  \Gamma = \dispyng{4} \oplus \dispyng{3,1} \oplus \dispyng{2,2}. \label{eq:221decomposition}
\end{equation}
This can also be verified using orthogonality of characters~\cite{Chen1989}. This decomposition is also identical to the decomposition of the tensor product $(2,0,0)\otimes(2,0,0)$ of the $\mathrm{su}(4)$ irrep; the irrep $(2,0,0)$ is appropriate to describe the possible states of two indistinguishable photons, as required by our detection
scheme where output photons in a given output port are indistinguishable~\cite{Chen1989}.  The Young tableau manipulations of
Eq.(\eqref{eq:semistandardtableaus}) are simply an application of the well-known Littlewood-Richardson rule for decomposing
tensor products of $\mathrm{su}(n)$ irreps. The dimension of each irrep $\lambda$ is the number of standard Young tableaux that can be constructed from the Young diagram $\lambda$. For the irreps that occur in~\eqref{eq:221decomposition} the standard Young tableaux are
\begin{equation}
  {\Yboxdim{10pt} \young(1234), \qquad \young(123,4), \qquad \young(124,3), \qquad \young(134,2), \qquad \young(12,34), \qquad \young(13,24).} \label{eq:irrepdims}
\end{equation}
Based on these constructions we can infer that the irrep $\inlineyng{4}$ is one dimensional, the irrep $\inlineyng{3,1}$ is three dimensional and the irrep $\inlineyng{2,2}$ is two dimensional. The sum of these dimensions is six, the same as the size of the rate matrix and the dimension of $\Gamma$. 

We now find the basis which reduces $\Gamma$ to block-diagonal form. We use the method of class-operators. The basis vectors for an irrep are the linear combinations of the eigenvectors of the complete set of commuting operators (CSCO). A CSCO of symmetric group $S_n$ is formed from the set $\set{D_k^{(2)}}_{k=2}^n$ of two-cycle class operators~\cite{Chen1989},
\begin{equation}
  D_k^{(2)} = \sum_{\sigma^{(2)}\in S_k} \sigma^{(2)}, \label{eq:classoperatordef}
\end{equation}
of the canonical subgroup chain $S_n \supset S_{n-1} \cdots \supset S_2$. Using the representations of the two cycles we can construct the two-cycle class operators,
\begin{align}
  D_2(2) &=
  \begin{pmatrix}
    1 & 0 & 0 & 0 & 0 & 0 \\
    0 & 0 & 0 & 1 & 0 & 0 \\
    0 & 0 & 0 & 0 & 1 & 0 \\
    0 & 1 & 0 & 0 & 0 & 0 \\
    0 & 0 & 1 & 0 & 0 & 0 \\
    0 & 0 & 0 & 0 & 0 & 1
  \end{pmatrix}, \quad  D_3(2) = \begin{pmatrix}
    1 & 1 & 0 & 1 & 0 & 0 \\
    1 & 1 & 0 & 1 & 0 & 0 \\
    0 & 0 & 1 & 0 & 1 & 1 \\
    1 & 1 & 0 & 1 & 0 & 0 \\
    0 & 0 & 1 & 0 & 1 & 1 \\
    0 & 0 & 1 & 0 & 1 & 1
  \end{pmatrix}, \quad  D_4(2) = \begin{pmatrix}
    2 & 1 & 1 & 1 & 1 & 0 \\
    1 & 2 & 1 & 1 & 0 & 1 \\
    1 & 1 & 2 & 0 & 1 & 1 \\
    1 & 1 & 0 & 2 & 1 & 1 \\
    1 & 0 & 1 & 1 & 2 & 1 \\
    0 & 1 & 1 & 1 & 1 & 2
  \end{pmatrix}.
\end{align}
Next, the eigenbasis $V$ that simultaneously diagonalizes the operators $\set{D_k^{(2)}}_{k=2}^n$ is found. The eigenbasis can be found using numerical simultaneous diagonalization algorithms~\cite{Gerstner1993}. However, for small $n$, a simpler procedure is as follows. The two-cycle class operators have eigenvalues
\begin{equation}
  \kappa_\lambda^{(2)} = \frac{n}{2} + \frac{1}{2}\sum_{\ell=1}^m \lambda_\ell(\lambda_\ell-2\ell), \label{eq:classoperatoreig}
\end{equation}
where each eigenvalue is labeled by a Young diagram $\lambda$ and $\lambda_\ell$ is the number of boxes on the $\ell$-th row of the Young diagram $\lambda$.  The chains for $S_4$ are 
\begin{equation}
  \begingroup
\renewcommand*{\arraystretch}{1.5}
  \Yboxdim{10pt}
  \begin{matrix}
     \young(1234) & \to & \young(123) & \to & \young(12),      \\[0.3em]
    (6)       &     & (3)      &     & (1)       \\[0.5em]
    \young(123,4) & \to & \young(123) & \to & \young(12),    \\[0.3em]
    (2)       &     & (3)      &     & (1)       \\[0.5em]
    \young(124,3) & \to & \young(12,3) & \to & \young(12),  \\[0.3em]
    (2)       &     & (0)      &     & (1)       \\[0.5em]
    \young(134,2) & \to & \young(13,2) & \to & \young(1,2),\\[0.3em]
    (2)       &     & (0)      &     & (-1)      \\[0.5em]
    \young(12,34) & \to & \young(12,3) & \to & \young(12),  \\[0.3em]
    (0)       &     & (0)      &     & (1)       \\[0.5em]
    \young(13,24) & \to & \young(13,2) & \to & \young(1,2),\\[0.3em]
    (0)       &     & (0)      &     & (-1) 
  \end{matrix}
  \endgroup
\end{equation}
where the eigenvalue associated with each Young diagram via Eq.~\eqref{eq:classoperatoreig} is given below it. Each of the chains of Young diagrams is associated with one simultaneous eigenvector of $\set{D_k^{(2)}}_{k=2}^n$. For instance, the first eigenvector has eigenvalue $6$ for $D_4^{(2)}$, eigenvalue $3$ for $D_3^{(2)}$ and eigenvalue $1$ for $D_2^{(2)}$. 

A simple way of finding the eigenbasis is to use the operator 
\begin{equation}
  D^{(2)} = \sum_{k=2}^4 \alpha_i D_k^{(2)},
\end{equation}
where the constants $\alpha_i$ are chosen so that $D^{(2)}$ has no repeated eigenvalues. For $n \le 6$ the choice of $\alpha_i = i+7$ works. With this choice
\begin{equation}
  D^{(2)} =
  \begin{pmatrix}
    41 & 21 & 11 & 21 & 11 & 0 \\
    21 & 32 & 11 & 30 & 0 & 11 \\
    11 & 11 & 32 & 0 & 30 & 21 \\
    21 & 30 & 0 & 32 & 11 & 11 \\
    11 & 0 & 30 & 11 & 32 & 21 \\
    0 & 11 & 21 & 11 & 21 & 41
  \end{pmatrix}.
\end{equation}
The eigenbasis of $D^{(2)}$ is
\begin{equation}
    \begingroup
\renewcommand*{\arraystretch}{1.5}
  V = \begin{pmatrix}
    \frac{1}{\sqrt{6}} & \frac{1}{\sqrt{6}} & \frac{1}{\sqrt{6}} & \frac{1}{\sqrt{6}} & \frac{1}{\sqrt{6}} & \frac{1}{\sqrt{6}} \\
    0 & -\frac{1}{2} & -\frac{1}{2} & \frac{1}{2} & \frac{1}{2} & 0 \\
    -\frac{1}{\sqrt{3}} & \frac{1}{2 \sqrt{3}} & -\frac{1}{2 \sqrt{3}} & \frac{1}{2 \sqrt{3}} & -\frac{1}{2 \sqrt{3}} & \frac{1}{\sqrt{3}} \\
    \frac{1}{\sqrt{6}} & \frac{1}{\sqrt{6}} & -\frac{1}{\sqrt{6}} & \frac{1}{\sqrt{6}} & -\frac{1}{\sqrt{6}} & -\frac{1}{\sqrt{6}}  \\
    0 & \frac{1}{2} & -\frac{1}{2} & -\frac{1}{2} & \frac{1}{2} & 0 \\
    \frac{1}{\sqrt{3}} & -\frac{1}{2 \sqrt{3}} & -\frac{1}{2 \sqrt{3}} & -\frac{1}{2 \sqrt{3}} & -\frac{1}{2 \sqrt{3}} & \frac{1}{\sqrt{3}}
  \end{pmatrix}, \label{eq:basistransformation22}
  \endgroup
\end{equation}
where the vectors (rows) have been ordered to match the ordering of the irreps in~\eqref{eq:221decomposition}. In this basis the rate matrix $VR(\vec{\tau})V^{\mathrm{T}}$ is block-diagonal. 

We now analyze the effects of the basis change $V$ on the interferometer matrix $\vec{u}$. Using Eqs.~\eqref{eq:interferometervectorexample} and~\eqref{eq:basistransformation22}, we find
\begin{equation}
  V\vec{u} = \frac{1}{\sqrt{6}}
  \begin{pmatrix}
    U_{23}U_{33}U_{12}^2 + 2U_{13}U_{23}U_{32}U_{12} + 2U_{13}U_{22}U_{33}U_{12} + U_{13}^2U_{22}U_{32} \\
    0 \\
    -\sqrt{2}U_{23}U_{33}U_{12}^2 - \sqrt{2}U_{13}U_{23}U_{32}U_{12} + \sqrt{2}U_{13}U_{22}U_{33}U_{12} + \sqrt{2}U_{13}^2U_{22}U_{32} \\
    -U_{23}U_{33}U_{12}^2 + 2U_{13}U_{23}U_{32}U_{12} - 2U_{13}U_{22}U_{33}U_{12} + U_{13}^2U_{22}U_{32} \\
    0 \\
    \sqrt{2}U_{23}U_{33}U_{12}^2 - \sqrt{2}U_{13}U_{23}U_{32}U_{12} - \sqrt{2}U_{13}U_{22}U_{33}U_{12} + \sqrt{2}U_{13}^2U_{22}U_{32} 
 \end{pmatrix}. \label{eq:rotatedinterferometervectorexample}
\end{equation}

The entries of $V\vec{u}$ can be expressed as immanants. Immanants are defined using the characters of the symmetric group. For $S_4$ the relevant characters are provided in Table~\ref{tab:charactertableS4}.
\begin{table}[h]
  \centering
  \begin{tabular}{c| c c c c c}
    Irrep $\lambda$/Character & $\chi^\lambda\left(\inlineyng{4}\right)$ & $\chi^\lambda\left(\inlineyng{3,1}\right)$ & $\chi^\lambda\left(\inlineyng{2,2}\right)$ & $\chi^\lambda\left(\inlineyng{2,1,1}\right)$ & $\chi^\lambda\left(\inlineyng{1,1,1,1}\right)$  \\\hline
    $\inlineyng{4}$ & $1$ & $1$ & $1$ & $1$ & $1$ \\
    $\inlineyng{3,1}$ & $-1$ & $0$ & $-1$ & $1$ & $3$ \\
    $\inlineyng{2,2}$ & $0$ & $-1$ & $2$ & $0$ & $2$
  \end{tabular}
  \caption{Characters for three irreps of $S_4$. Each row provides the character of one irrep, where each column entry is the character of one class of $S_4$.}
  \label{tab:charactertableS4}
\end{table}

The immanants are functions of the scattering matrix $T$ or the matrix found by permuting the rows of $T$. Using the definition of the scattering matrix~\eqref{eq:scatteringmatrixdef}, where $\vec{\upsilon}=1123$ and $\vec{\xi}=2233$, the matrix 
\begin{equation}
  T =
  \begin{pmatrix}
    U_{12} & U_{12} & U_{13} & U_{13} \\
    U_{12} & U_{12} & U_{13} & U_{13} \\
    U_{22} & U_{22} & U_{23} & U_{23} \\
    U_{32} & U_{32} & U_{33} & U_{33}
  \end{pmatrix}.
\end{equation}
Using Eq.~\eqref{eq:immanantdef} and Table~\ref{tab:charactertableS4}, we find that this matrix has associated immanants
\begin{align}
  \imm{4}T &= 4 U_{23} U_{33} U_{12}^2+8 U_{13} U_{23} U_{32} U_{12}+8 U_{13} U_{22} U_{33} U_{12}+4U_{13}^2 U_{22} U_{32}, \\
  \imm{3,1}T &= 4 U_{12}^2 U_{23} U_{33}-4 U_{13}^2 U_{22} U_{32}, \\
  \imm{3,1}T_{(13)} &= 4 U_{12} U_{13} U_{22} U_{33}-4 U_{12} U_{13} U_{23} U_{32}, \\
  \imm{2,2}T &= 4 U_{23} U_{33} U_{12}^2-4 U_{13} U_{23} U_{32} U_{12}-4 U_{13} U_{22} U_{33} U_{12}+4
   U_{13}^2 U_{22} U_{32},
\end{align}
where we remind the reader that $T_{(13)}$ is the matrix $T$ with its first and third rows permuted. This is the complete list of distinct immanants of types $\inlineyng{4}$, $\inlineyng{3,1}$ and $\inlineyng{2,2}$ that can be formed from $T_\sigma$ for all $\sigma \in S_4$.

To transform the entries of $V\vec{u}$ into immanants, we require the help of Eq.~\eqref{eq:irrepdims} which specifies the sizes of each block of the rotated interferometer vector $V\vec{u}$. We note that the first block associated with irrep $\inlineyng{4}$ is of size $1$. The first entry in Eq.~\eqref{eq:rotatedinterferometervectorexample}, must then be the permanent $\imm{4}T$ of the scattering matrix up to a constant. We can verify that this is so. The next block, associated with irrep $\inlineyng{3,1}$ is three-dimensional. We should able to express the second to fifth entries of Eq.~\eqref{eq:rotatedinterferometervectorexample} in terms of the two immanants $\imm{3,1}T$ and $\imm{3,1}T_{(13)}$. The simple rotation,
\begin{equation}
\frac{1}{4\sqrt{6}}\begin{pmatrix}
    0 & 0 \\
    -\sqrt{2} & \sqrt{2} \\
    -1 & -2
  \end{pmatrix}
  \begin{pmatrix}
    \imm{3,1}T \\ \imm{3,1}T_{(13)}
  \end{pmatrix},
\end{equation}
is indeed equal to the second to fifth entries of Eq.~\eqref{eq:rotatedinterferometervectorexample}. The same procedure can be followed for the final two-dimensional block associated with the irrep $\inlineyng{2,2}$. Hence, in the basis $V$ the entries of the interferometer vector,
\begin{equation}
  V\vec{u} =
  \frac{1}{4\sqrt{6}}\begin{pmatrix}
    \imm{4}T \\
    0 \\
    -\sqrt{2}\imm{3,1}T +\sqrt{2}\imm{3,1}T_{(13)} \\
    -\imm{3,1}T - 2\imm{3,1}T_{(13)} \\
    0\\
    \sqrt{2}\imm{2,2}T
  \end{pmatrix}, \label{eq:interferometervectorimmanants}
\end{equation}
are transformed to immanants.

Finally, using the rate expression~\eqref{eq:rotatedrate}, we find that
\begin{align}
  C(\vec{\tau}) &= \frac{1}{96}(\Delta_{\e} + \Delta_{(12)} + \Delta_{(123)} + \Delta_{(12,34)} + \Delta_{(1234)} + \Delta_{(124)} + \Delta_{(1243)} + \Delta_{(13)} + \Delta_{(132)} \nonumber\\&\quad+ \Delta_{(13,24)} + \Delta_{(1324)} + \Delta_{(134)} + \Delta_{(1342)} + \Delta_{(14)} + \Delta_{(142)} + \Delta_{(14,23)} + \Delta_{(1423)} \nonumber\\&\quad+ \Delta_{(143)} + \Delta_{(1432)} + \Delta_{(23)} + \Delta_{(234)} + \Delta_{(24)} + \Delta_{(243)} + \Delta_{(34)})|\imm{4}T|^2 \nonumber\\&\quad+ \frac{1}{32}(\Delta_{\e} + \Delta_{(12)} + \Delta_{(12,34)} - \Delta_{(13,24)} - \Delta_{(1324)} - \Delta_{(14,23)} - \Delta_{(1423)} + \Delta_{(34)})|\imm{3,1}T|^2 \nonumber\\&\quad +  \frac{1}{96}(2\Delta_{\e} + 2\Delta_{(12)} + \Delta_{(123)} - 2\Delta_{(12,34)} - \Delta_{(1234)} + \Delta_{(124)} - \Delta_{(1243)} + \Delta_{(13)} + \Delta_{(132)} - \Delta_{(134)} \nonumber\\&\quad - \Delta_{(1342)} + \Delta_{(14)} + \Delta_{(142)} - \Delta_{(143)} - \Delta_{(1432)} + \Delta_{(23)} - \Delta_{(234)} + \Delta_{(24)} - \Delta_{(243)} - 2\Delta_{(34)})|\imm{3,1}T_{(13)}|^2 \nonumber\\&\quad+ \frac{1}{96}(\Delta_{(123)} - \Delta_{(1234)} - \Delta_{(124)} + \Delta_{(1243)} + \Delta_{(13)} + \Delta_{(132)} - \Delta_{(134)} - \Delta_{(1342)} - \Delta_{(14)} \nonumber\\&\quad- \Delta_{(142)} + \Delta_{(143)} + \Delta_{(1432)} + \Delta_{(23)} - \Delta_{(234)} - \Delta_{(24)} + \Delta_{(243)})(\imm{3,1}T) (\imm{3,1}T_{(13)})^* \nonumber\\&\quad+ \frac{1}{96}(\Delta_{(123)} + \Delta_{(1234)} - \Delta_{(124)} - \Delta_{(1243)} + \Delta_{(13)} + \Delta_{(132)} + \Delta_{(134)} + \Delta_{(1342)} - \Delta_{(14)} \nonumber\\&\quad- \Delta_{(142)} - \Delta_{(143)} - \Delta_{(1432)} + \Delta_{(23)} + \Delta_{(234)} - \Delta_{(24)} - \Delta_{(243)}) (\imm{3,1}T_{(13)}) (\imm{3,1}T)^* \nonumber\\&\quad + \frac{1}{96}(2\Delta_{\e} + 2\Delta_{(12)} - \Delta_{(123)} + 2\Delta_{(12,34)} - \Delta_{(1234)} - \Delta_{(124)} - \Delta_{(1243)} - \Delta_{(13)} - \Delta_{(132)} + 2\Delta_{(13,24)} \nonumber\\&\quad + 2\Delta_{(1324)} - \Delta_{(134)} - \Delta_{(1342)} - \Delta_{(14)} - \Delta_{(142)} + 2\Delta_{(14,23)} + 2\Delta_{(1423)} - \Delta_{(143)} - \Delta_{(1432)} - \Delta_{(23)} \nonumber\\&\quad- \Delta_{(234)} - \Delta_{(24)} - \Delta_{(243)} + 2\Delta_{(34)})|\imm{2,2}T|^2. \label{eq:examplerate}
\end{align}

The coincidence landscape shown in Fig.~\ref{fig:landscape211} is plotted using the interferometer matrix
\begin{equation}
  U = \begin{pmatrix}
     \phantom{+}0.232231+0.437219\iota & -0.271046+0.371938\iota & \phantom{+}0.168757+0.717374\iota \\
 -0.430781+0.406851\iota & -0.447972-0.0160114\iota & \phantom{+}0.539331-0.396341\iota \\
 -0.170262-0.612224\iota & -0.476765+0.599963\iota & -0.0840731-0.043172\iota \\
  \end{pmatrix}. \label{eq:plotmatrix}
\end{equation}

\section{Hong-Ou-Mandel Corrections}
\label{sec:appendixhom}
In this section we present the corrected rate for the HOM expression. For input state $\ket{11;12;\vec{\tau}}$ and measurement operator $M_{11}$, the scattering matrix
\begin{equation}
  T^{12,12} =
  \begin{pmatrix}
    U_{11} & U_{12} \\
    U_{21} & U_{22} 
  \end{pmatrix}.
\end{equation}
This matrix has two non-trivial immanants; the permanent,
\begin{equation}
  \imm{2}T = U_{12}U_{21}+U_{11}U_{22}, 
\end{equation}
and the determinant,
\begin{equation}
 \imm{1,1}T = U_{11} U_{22}-U_{12} U_{21}.
\end{equation}
The coincidence rate is given by
\begin{align}
C^{12,12}(\vec{\tau}) &=  \left(\Delta_{\e}+ \Delta_{(12)}\right)\left|\imm{2}T^{12,12}\right|^2 + \left(\Delta_{\e}-\Delta_{(12)}\right)\left|\imm{1,1}T^{12,12}\right|^2,
\end{align}
where
\begin{align}
  \Delta_{\e} &= \int\dw{}_1\dw{}_2 |\phi(\omega_1)|^2|\phi(\omega_2)|^2,\\
  \Delta_{(12)} &= \int\dw{}_1\dw{}_2 \left|\phi(\omega_1)\right|^2\left|\phi(\omega_2)\right|^2\e^{\iota\omega_1(\tau_1-\tau_2)}\e^{\iota\omega_2(\tau_2-\tau_1)}.
\end{align}

For input state $\ket{22;1122;\vec{\tau}}$ and measurement operator $M_{22}$, the scattering matrix
\begin{equation}
  T^{1122,1122} =
  \begin{pmatrix}
    U_{11} & U_{11} & U_{12} & U_{12} \\
    U_{11} & U_{11} & U_{12} & U_{12} \\
    U_{21} & U_{21} & U_{22} & U_{22} \\
    U_{21} & U_{21} & U_{22} & U_{22} \\
  \end{pmatrix}.
\end{equation}
This matrix has three distinct immanants of the $\inlineyng{4},\inlineyng{3,1},\inlineyng{2,2}$ types,
\begin{align}
  \imm{4}T^{1122,1122} &= 4U_{12}^2U_{21}^2 + 16U_{11}U_{12}U_{22}U_{21} + 4U_{11}^2U_{22}^2, \\
  \imm{3,1}T^{1122,1122} &= 4U_{11}^2U_{22}^2 - 4U_{12}^2U_{21}^2, \\
  \imm{2,2}T^{1122,1122} &= 4U_{12}^2U_{21}^2 - 8U_{11}U_{12}U_{22}U_{21} + 4U_{11}^2U_{22}^2.
\end{align}
The coincidence rate is given by
\begin{align}
  C^{1122,1122}(\vec{\tau}) &= \frac{1}{96}(\Delta_{\e}+\Delta_{(12)}+\Delta_{(123)}+\Delta_{(12,34)}+\Delta_{(1234)}+\Delta_{(124)}+\Delta_{(1243)}+\Delta_{(13)}+\Delta_{(132)}+\Delta_{(13,24)}\nonumber\\&\qquad+\Delta_{(1324)}+\Delta_{(134)}+\Delta_{(1342)} +\Delta_{(14)}+\Delta_{(142)}+\Delta_{(14,23)}+\Delta_{(1423)}+\Delta_{(143)}+\Delta_{(1432)}\nonumber\\&\quad+\Delta_{(23)}+\Delta_{(234)}+\Delta_{(24)}+\Delta_{(243)}+\Delta_{(34)})\left|\imm{4}T^{1122,1122}\right|^2 + \frac{1}{32}(\Delta_{\e}+\Delta_{(12)}+\Delta_{(12,34)}\nonumber\\&\quad-\Delta_{(13,24)}-\Delta_{(1324)}-\Delta_{(14,23)}-\Delta_{(1423)}+\Delta_{(34)})\left|\imm{3,1}T^{1122,1122}\right|^2 + \frac{1}{96}(2\Delta_{\e}+2\Delta_{(12)}\nonumber\\&\quad-\Delta_{(123)}+2\Delta_{(12,34)}-\Delta_{(1234)}-\Delta_{(124)}-\Delta_{(1243)}-\Delta_{(13)}-\Delta_{(132)}+2\Delta_{(13,24)}+2\Delta_{(1324)}\nonumber\\&\qquad-\Delta_{(134)}-\Delta_{(1342)}-\Delta_{(14)}-\Delta_{(142)}+2\Delta_{(14,23)}+2\Delta_{(1423)}-\Delta_{(143)}-\Delta_{(1432)}-\Delta_{(23)}\nonumber\\&\qquad-\Delta_{(234)}-\Delta_{(24)}-\Delta_{(243)}+2\Delta_{(34)})\left|\imm{2,2}T^{1122,1122}\right|^2,
\end{align}
where
\begin{equation}
  \Delta_\sigma = \int \dw{}_1\dw{}_2\dw{}_3\dw{}_4 |\phi(\omega_1)\phi(\omega_2)\phi(\omega_3)\phi(\omega_4)|^2P_\sigma\e^{-\iota(\omega_1\tau_1+\omega_2\tau_2+\omega_3\tau_3+\omega_4\tau_4)} \label{eq:deltadef}
\end{equation}

For input state $\ket{22;1122;\vec{\tau}}$ and measurement operator $M_{31}$, the scattering matrix
\begin{equation}
  T^{1122,1112} =
  \begin{pmatrix}
    U_{11} & U_{11} & U_{11} & U_{12} \\
    U_{11} & U_{11} & U_{11} & U_{12} \\
    U_{21} & U_{21} & U_{21} & U_{22} \\
    U_{21} & U_{21} & U_{21} & U_{22} \\
  \end{pmatrix}.
\end{equation}
This matrix has two distinct immanants of the $\inlineyng{4},\inlineyng{3,1}$ types,
\begin{align}
  \imm{4}T^{1122,1112} &= 12U_{21}U_{22}U_{11}^2 + 12U_{12}U_{21}^2U_{11}, \\
  \imm{3,1}T^{1122,1112} &=  4U_{11}^2U_{21}U_{22} - 4U_{11}U_{12}U_{21}^2.
\end{align}
The coincidence rate is given by
\begin{align}
  C^{1122,1112}(\vec{\tau}) &= \frac{1}{144}(\Delta_{\e} + \Delta_{(12)} + \Delta_{(123)} + \Delta_{(12,34)} + \Delta_{(1234)} + \Delta_{(124)} + \Delta_{(1243)} + \Delta_{(13)}  + \Delta_{(132)} \nonumber\\& \qquad+ \Delta_{(13,24)} + \Delta_{(1324)} + \Delta_{(134)} + \Delta_{(1342)} + \Delta_{(14)} + \Delta_{(142)} + \Delta_{(14,23)} + \Delta_{(1423)} \nonumber\\& \qquad+ \Delta_{(143)} + \Delta_{(1432)} + \Delta_{(23)} + \Delta_{(234)} + \Delta_{(24)}  + \Delta_{(243)} + \Delta_{(34)})\left|\imm{4}T^{1122,1112}\right|^2  \nonumber\\& \qquad+ \frac{1}{16}(\Delta_{\e} + \Delta_{(12)} + \Delta_{(12,34)} - \Delta_{(13,24)} - \Delta_{(1324)} - \Delta_{(14,23)} - \Delta_{(1423)} \nonumber\\& \qquad+ \Delta_{(34)})\left|\imm{3,1}T^{1122,1112}\right|^2,
\end{align}
where $\Delta_\sigma$ is defined in Eq.~\eqref{eq:deltadef}.

For input state $\ket{22;1122;\vec{\tau}}$ and measurement operator $M_{13}$, the scattering matrix
\begin{equation}
  T^{1122,1222} =
  \begin{pmatrix}
    U_{11} & U_{12} & U_{12} & U_{12} \\
    U_{11} & U_{12} & U_{12} & U_{12} \\
    U_{21} & U_{22} & U_{22} & U_{22} \\
    U_{21} & U_{22} & U_{22} & U_{22} \\
  \end{pmatrix}.
\end{equation}
This matrix has two distinct immanants of the $\inlineyng{4},\inlineyng{3,1}$ types,
\begin{align}
  \imm{4}T^{1122,1222} &= 12U_{21}U_{22}U_{12}^2 + 12U_{11}U_{22}^2U_{12}, \\
  \imm{3,1}T^{1122,1222} &= 4U_{11}U_{12}U_{22}^2 - 4U_{12}^2U_{21}U_{22}.
\end{align}
The coincidence rate is given by
\begin{align}
  C^{1122,1222}(\vec{\tau}) &= \frac{1}{144}(\Delta_{\e} + \Delta_{(12)} + \Delta_{(123)} + \Delta_{(12,34)} + \Delta_{(1234)} + \Delta_{(124)} + \Delta_{(1243)} + \Delta_{(13)} + \Delta_{(132)} \nonumber\\&\qquad + \Delta_{(13,24)} + \Delta_{(1324)} + \Delta_{(134)} + \Delta_{(1342)} + \Delta_{(14)} + \Delta_{(142)} + \Delta_{(14,23)} + \Delta_{(1423)} \nonumber\\&\qquad + \Delta_{(143)} + \Delta_{(1432)} + \Delta_{(23)} + \Delta_{(234)} + \Delta_{(24)} + \Delta_{(243)} + \Delta_{(34)})\left|\imm{4}T^{1122,1222}\right|^2 \nonumber\\&\qquad + \frac{1}{16}(\Delta_{\e} + \Delta_{(12)} + \Delta_{(12,34)} - \Delta_{(13,24)} - \Delta_{(1324)} - \Delta_{(14,23)} - \Delta_{(1423)} \nonumber\\&\qquad+ \Delta_{(34)})\left|\imm{3,1}T^{1122,1222}\right|^2,
\end{align}
where $\Delta_\sigma$ is defined in Eq.~\eqref{eq:deltadef}.

The two pairs of photons are completely distinguishable in the limit $\tau_2-\tau_2\to\infty$. In this case, the difference
\begin{align}
  h_c &= \frac{p^2}{384}\bigg(4\left|\imm{4}T^{1122,1122}\right|^2 + 8\left|\imm{2,2}T^{1122,1122}\right|^2 + (3+2\sqrt{2})\left|\imm{3,1}T^{1122,1122}\right|^2 \nonumber\\&\qquad+ 
4\left|\imm{4}T^{1122,1112}\right|^2 + 4(3+2\sqrt{2})\left|\imm{3,1}T^{1122,1112}\right|^2 \nonumber\\&\qquad+ 4\left|\imm{4}T^{1122,1222}\right|^2 + 4(3+2\sqrt{2})\left|\imm{3,1}T^{1122,1222}\right|^2\bigg).
\end{align}

\end{widetext}
\end{document}